\documentclass[paper=letter,fontsize=11]{scrartcl}
\pdfoutput=1
\setkomafont{disposition}{\bfseries}

\usepackage[utf8]{inputenc}
\usepackage{lmodern}
\usepackage{amsmath, amsthm, amssymb, mathtools}
\usepackage{amsfonts}
\usepackage{bm}
\usepackage{booktabs}
\usepackage{braket}
\usepackage{xcolor}
\usepackage{enumitem}
\usepackage{etoolbox}
\usepackage{inconsolata}
\usepackage[left=1in,top=1in,right=1in,bottom=1in]{geometry} % better than fullpage
\usepackage{graphicx}
\usepackage{mleftright}\mleftright
\usepackage{caption}
\usepackage{microtype}
\PassOptionsToPackage{hyphens}{url}
\usepackage[pagebackref]{hyperref}
\usepackage{qcircuit}
\usepackage{thmtools,thm-restate}
\usepackage{url}
\usepackage[normalem]{ulem}
\usepackage{verbatim}
\usepackage[linesnumbered,ruled,procnumbered]{algorithm2e}
\usepackage[nameinlink,capitalize]{cleveref}

\definecolor{linkc}{HTML}{0030c0}
\hypersetup{
    pdftitle={Query-optimal estimation of unitary channels in diamond distance}, % title
    pdfauthor={Jeongwan Haah, Robin Kothari, Ryan O'Donnell, and Ewin Tang}, % authors
    colorlinks=true, % false: boxed links; true: colored links
    linkcolor=linkc, % color of internal links
    citecolor=linkc, % color of links to bibliography
    urlcolor=linkc % color of external links
}

\renewcommand{\backref}[1]{}

\renewcommand{\backrefalt}[4]{% 
\ifcase #1 % 
\or 
[p.\ #2]% 
\else 
[pp.\ #2]% 
\fi}

\declaretheorem[numberwithin=section]{theorem}
\declaretheorem[sibling=theorem]{lemma}

\declaretheorem[sibling=theorem,name=Proposition]{proposition}
\theoremstyle{definition}

\declaretheorem[sibling=theorem]{definition}
\declaretheorem[sibling=theorem]{remark}

\makeatletter
\newcommand{\para}{%
    \@startsection{paragraph}{4}%
    {\z@}{2ex \@plus 3.3ex \@minus .2ex}{-1em}%
    {\normalfont\normalsize\bfseries}%
}
\makeatother

\DeclarePairedDelimiter{\ceil}{\lceil}{\rceil}
\DeclarePairedDelimiter{\floor}{\lfloor}{\rfloor}
\DeclarePairedDelimiter{\abs}{\lvert}{\rvert} % use \abs* to do \left|#1\right|
\DeclarePairedDelimiter{\opnorm}{\lVert}{\rVert_{\mathrm{op}}}
\DeclarePairedDelimiter{\diamnorm}{\lVert}{\rVert_{\diamond}}
\DeclarePairedDelimiter{\onenorm}{\lVert}{\rVert_{1}}
\DeclarePairedDelimiter{\twonorm}{\lVert}{\rVert_{2}}

\DeclarePairedDelimiter{\fnorm}{\lVert}{\rVert_{\textup{F}}}

\DeclarePairedDelimiter{\bracks}{\lbrack}{\rbrack}
\DeclarePairedDelimiter{\parens}{\lparen}{\rparen}

\DeclareMathOperator{\poly}{poly}
\DeclareMathOperator{\polylog}{polylog}

\DeclareMathOperator*{\Tr}{Tr}

\DeclareMathOperator{\Real}{Re}
\DeclareMathOperator{\Imag}{Im}

\newcommand{\dH}[2]{d_{\widetilde{\mathrm{H}}}(#1,#2)}
\DeclareMathOperator*{\Ex}{\mathbf{E}} \DeclareMathOperator*{\E}{\mathbf{E}} 

\DeclareMathOperator{\dist}{dist}

\DeclareMathOperator{\emeasure}{error} % for a generic error measure

\DeclareMathOperator{\diag}{diag}

\newcommand{\eps}{\varepsilon}
\newcommand{\diff}{\partial}

\newcommand{\ZZ}{{\mathbb{Z}}}

\newcommand{\RR}{{\mathbb{R}}}
\newcommand{\CC}{{\mathbb{C}}}

\newcommand{\ii}{\mathrm{i}}

\newcommand{\bb}{{\boldsymbol{b}}}
\newcommand{\bD}{{\boldsymbol{D}}}
\newcommand{\bE}{{\boldsymbol{E}}}
\newcommand{\bG}{{\boldsymbol{G}}}

\newcommand{\bJ}{{\boldsymbol{J}}}
\newcommand{\bL}{{\boldsymbol{L}}}
\newcommand{\bM}{{\boldsymbol{M}}}
\newcommand{\bu}{{\boldsymbol{u}}}
\newcommand{\bU}{{\boldsymbol{U}}}
\newcommand{\bv}{{\boldsymbol{v}}}
\newcommand{\bV}{{\boldsymbol{V}}}
\newcommand{\bw}{{\boldsymbol{w}}}
\newcommand{\bW}{{\boldsymbol{W}}}
\newcommand{\bx}{{\boldsymbol{x}}}
\newcommand{\bX}{{\boldsymbol{X}}}
\newcommand{\by}{{\boldsymbol{y}}}

\newcommand{\bY}{{\boldsymbol{Y}}}
\newcommand{\bz}{{\boldsymbol{z}}}
\newcommand{\bZ}{{\boldsymbol{Z}}}
\newcommand{\bphi}{{\boldsymbol{\phi}}}
\newcommand{\bPhi}{{\boldsymbol{\Phi}}}
\newcommand{\bpsi}{{\boldsymbol{\psi}}}
\newcommand{\brho}{{\boldsymbol{\rho}}}
\newcommand{\bSigma}{{\boldsymbol{\Sigma}}}
\newcommand{\bPsi}{{\boldsymbol{\Psi}}}
\newcommand{\bdelta}{{\boldsymbol{\delta}}}
\newcommand{\bDelta}{{\boldsymbol{\Delta}}}
\newcommand{\beps}{{\boldsymbol{\eps}}}
\newcommand{\blambda}{{\boldsymbol{\lambda}}}
\newcommand{\wh}[1]{\widehat{#1}}

\newcommand{\calE}{\mathcal{E}}
\newcommand{\calM}{\mathcal{M}}
\newcommand{\calR}{\mathcal{R}}
\newcommand{\calS}{\mathcal{S}}
\newcommand{\calU}{\mathcal{U}}
\newcommand{\ol}[1]{\overline{#1}}

\newcommand{\rd}{{\mathrm{d}}}

\newcommand{\half}{{\tfrac 1 2}}

\newcommand{\calW}{{\mathcal{W}}}
\newcommand{\calC}{{\mathcal{C}}}

\definecolor{ewintodo}{HTML}{de024b}

\definecolor{RKcomment}{HTML}{FF8800}

\newcommand{\ignore}[1]{}

\newcommand{\dd}{{\mathsf{d}}} % dimension
\newcommand{\nn}{{\mathsf{n}}} % number of qubits
\newcommand{\UU}{{\mathsf{U}}} % unitary group
\newcommand{\uu}{{\mathfrak{u}}} % Lie algebra of the unitary group
\newcommand{\PPUU}{{\mathsf{PU}}} % projective unitary group
\newcommand{\uone}{{\UU(1)}} % the group of unimodular complex numbers.
\newcommand{\BB}{{\mathcal{B}}} % metric ball in the unitary group.
\newcommand{\PPBB}{{\mathcal{PB}}} % metric ball in the unitary group.
\newcommand{\net}{{\mathcal{N}}} % reflection net for the lower bound
\newcommand{\id}{{\mathbb{I}}} 
\newcommand{\pudist}[2]{\dist_{\mathrm{phaseop}}(#1,#2)}
\newcommand{\tvdist}[2]{\dist_{\mathrm{TV}}(#1,\,#2)}
\newcommand{\controlled}{{\mathrm{c}}} % controlled unitary
\DeclareMathOperator{\alg}{\mathcal{A}} % bootstrap algorithm
\DeclareMathOperator{\circuit}{\mathcal{C}} % bccks circuit
\DeclareMathOperator{\tomalg}{\mathcal{L}} % tomography algorithm
\newcommand{\macheps}{\varepsilon_{\mathrm{mach}}} % machine epsilon
\newcommand{\anc}{\textup{anc}} % ancilla state used in lower bd

\begin{document}
\title{Query-optimal estimation of\\ unitary channels in diamond distance}
\author{
Jeongwan Haah\footnote{Microsoft Quantum, Redmond, Washington, USA},~ 
Robin Kothari\footnote{Most of this work was done while the author was at Microsoft Quantum, Redmond, Washington, USA. The author is currently at Google Quantum AI, Venice, California, USA.},~ 
Ryan O'Donnell\footnote{Carnegie Mellon University Computer Science Department, Pittsburgh, Pennsylvania, USA. 
Part of this research was performed while the author was at Microsoft Quantum.},~
and Ewin Tang\footnote{University of Washington, Seattle, Washington, USA}
}

\maketitle

\begin{abstract}
We consider process tomography for unitary quantum channels.
Given access to an unknown unitary channel acting on a $\dd$-dimensional qudit, 
we aim to output a classical description of a unitary  
that is $\eps$-close to the unknown unitary in diamond norm.
We design an algorithm achieving error~$\eps$ 
using $O(\dd^2/\eps)$ applications of the unknown channel and only one qudit.
This improves over prior results, which use $O(\dd^3/\eps^2)$ [via standard process tomography] or $O(\dd^{2.5}/\eps)$ [Yang, Renner, and Chiribella, PRL 2020] applications.
To show this result, we introduce a simple technique to ``bootstrap'' an algorithm that can produce constant-error estimates 
to one that can produce $\eps$-error estimates 
with the Heisenberg scaling.
Finally, we prove a complementary lower bound showing that estimation requires $\Omega(\dd^2/\eps)$ applications, even with access to the inverse or controlled versions of the unknown unitary.
This shows that our algorithm has both optimal query complexity and optimal space complexity.
\end{abstract}

\begin{comment}
Unicode version of the abstract:

We consider process tomography for unitary quantum channels. Given access to an unknown unitary channel acting on a d-dimensional qudit, we aim to output a classical description of a unitary that is ε-close to the unknown unitary in diamond norm. We design an algorithm achieving error ε using O(d²/ε) applications of the unknown channel and only one qudit. This improves over prior results, which use O(d²/ε)⋅d/ε [via standard process tomography] or O(d²/ε)⋅√d [Yang, Renner, and Chiribella, PRL 2020] applications. To achieve this result, we introduce a simple technique to "bootstrap" an algorithm that can produce constant-error estimates to one that can produce ε-error estimates with the Heisenberg scaling. Finally, we prove a complementary lower bound showing that estimation requires Ω(d²/ε) applications, even with access to the inverse or controlled versions of the unknown unitary. This shows that our algorithm has both optimal query complexity and optimal space complexity.

Keywords:
process tomography
state tomography
universal programming
unitary learning
diamond norm
\end{comment}

\tableofcontents

\section{Introduction}

Quantum mechanics models the time evolution of a closed system
as a unitary operator on the vector space of states.
From the perspective of an experimentalist,
this model serves as a mechanism for producing statistical predictions.
This operational point of view prompts a fundamental question: when and how one can interact with a system with some unknown dynamics to learn its corresponding unitary operator?
This question generally goes under the name of \emph{process tomography},
and has been studied extensively in a variety of settings
depending on chosen figures of merit; 
see our discussion in~\cref{sec:CompPrior} below.
The figures of merit may be one or more of the following:
the accuracy of an estimate (measured with some metric
on the space of quantum channels),
the number of times one has to run the unknown dynamics,
the complexity of initial states and that of measurements,
the size of coherent Hilbert space required,
and the complexity of classical processing.

In this paper, we design an algorithm with the following properties. 
Let~$\diamnorm \cdot$ denote the diamond norm, 
the completely bounded norm on the space of linear operators acting on density matrices.
Let~$\UU(\dd)$ denote the set of all $\dd \times \dd$ unitary matrices.
If~$X \in \UU(\dd)$ is any unitary operator, 
we denote by~$\calU(X)$ the associated unitary channel
\begin{equation}    
    \calU(X) : \rho \mapsto X \rho X^\dagger.
\end{equation}

\begin{restatable}[Upper bound]{theorem}{upperbound}
    \label{thm:upperbound}
    There is a quantum algorithm~$\alg$ that,    
    given black-box access to an unknown $\dd$-dimensional unitary channel $Z \in \UU(\dd)$ and any $\eps > 0$, outputs a classical description of a unitary. This output is probabilistic and can be viewed as a $\UU(\dd)$-valued random variable $\wh{\bZ}$. The algorithm satisfies the following properties:
\begin{enumerate}
    \item \textup{(Query complexity)} $\alg$ queries the black box $O(\dd^2/\eps)$ times.
    \item \textup{(Space complexity)} $\alg$ only uses one qudit of dimension $\dd$. Specifically, it only prepares states of the form $V_2 (Z V_1)^p V_0 \ket 0$, where the positive integer~$p$ and unitaries~$V_0, V_1, V_2 \in \UU(\dd)$ are adaptively chosen, and measures them in the computational basis.
    \item \textup{(Gate complexity)} If the $\dd$-dimensional qudit is embedded in $\nn = \lceil \log_2 \dd \rceil$ qubits, then $\alg$ uses $\poly(\dd,1/\eps)$ one- and two-qubit quantum gates in total beyond the queries to the black box, and the classical time complexity is $\poly(\dd,1/\eps)$.
    \item \textup{(Closeness of output unitary)} The unitary output by $\alg$ is $\eps$-close to the unknown $Z$ in diamond norm with high probability. Specifically, $\E[~\diamnorm{\calU(\wh{\bZ})- \calU(Z)}^2~] \le \eps^2$, which implies $\Pr[~ \diamnorm{\calU(\wh{\bZ}) - \calU(Z)} \le 3\eps~] \ge \tfrac 2 3$.
    \item \textup{(Closeness of output channel)} We can also view the output of $\alg$ as a channel. Let $\calM$ denote the mixed unitary channel induced by running~$\alg$ on~$Z$ and then applying its output $\wh{\bZ}$. Then $\diamnorm*{\calM  - \calU(Z)} \le \eps^2$. \label{item:upperbound-channel}
\end{enumerate}    
\end{restatable}

\noindent
Complementary to these performance guarantees, we also show the following lower bound, which also holds against algorithms that can use $Z^\dagger$ and controlled versions of $Z$ and $Z^\dagger$.

\begin{restatable}[Lower bound]{theorem}{lowerbound}
    \label{thm:lowerbound}
    Let $\alg$ be an algorithm that, for an unknown $\dd$-dimensional unitary $Z \in \UU(\dd)$ accessible through black box oracles that implement $Z$, $Z^\dagger$, $\controlled Z = \ket 0 \bra 0 \otimes \id + \ket 1 \bra 1 \otimes Z$, and $\controlled Z^\dagger = \ket 0 \bra 0 \otimes \id + \ket 1 \bra 1 \otimes Z^\dag$, can output a classical description of a unitary channel $\calU(\wh{\bZ})$ such that $\diamnorm{\calU(\wh{\bZ}) - \calU(Z)} < \eps < \frac{1}{8}$ with probability $\geq \tfrac 2 3$.
    Then $\alg$ must use $\Omega(\frac{\dd^2}{\eps})$ oracle queries.
\end{restatable}

Due to subtle differences in assumptions and objectives,
we give comparison to the prior art 
after we properly define mathematical notions.
In brief: standard process tomography for general channels~\cite{leung2000towards} uses one qudit but requires a $1/\eps^2$ error dependence, which we improve quadratically.\footnote{We were not able to locate a $O(\dd^2/\eps^2)$-query process tomography algorithm for diamond-norm distance in the literature; we prove this result in \cref{sec:base}.}
Most prior work takes the accuracy measure to be either the one in \cref{item:upperbound-channel} of \cref{thm:upperbound}, referred to in this work as the problem of \emph{Approximate Storage-And-Retrieval}, or the closeness of the output unitary in \emph{entanglement infidelity}.
Prior work of Yang, Renner, and Chiribella~\cite{Yang2020} gives an algorithm that applies the unknown $Z$ in parallel to achieve the current best guarantees for these problems.
In particular, it recovers \cref{thm:upperbound}.\ref{item:upperbound-channel} with the same query complexity.
However, these problems are easier than that of unitary estimation to diamond-norm distance: norm conversion gives an algorithm for diamond-norm distance with $O(\dd^{2.5}/\eps)$ applications.
We improve on this by a factor of $O(\sqrt{\dd})$, and improve the number of qudits used from $O(\dd^{2.5}/\eps)$ to $1$.

\paragraph{Notation.}
Throughout the paper, 
all $\log$s have base $e = \sum_{k \ge 0} \frac 1 {k!} \approx 2.718$, and $\exp(x) = e^x = \sum_{j = 0}^\infty \frac{x^j}{j!}$ for any $x \in \RR$.
The dimension of a qudit is denoted by~$\dd \ge 1$, which we sometimes take to be $\nn$ qubits, so that $\dd = 2^\nn$.
For matrices, $\opnorm \cdot$ is the largest singular value (also known as the operator norm or spectral norm), $\onenorm \cdot$ is the trace norm (also known as the Schatten 1-norm), and $\fnorm \cdot$ is the Frobenius norm.
For a vector $v \in \CC^\dd$, $\diag(v) \in \CC^{\dd\times \dd}$ denotes the matrix with $v$ on the diagonal.
We use~$\id$ to denote the identity element of~$\UU(\dd)$.

\subsection{Preliminaries: Distances for unitary (and other) channels} \label{sec:norms}

We are mainly interested in distance measures between unitary channels.
However, for comparison with prior work we need to consider distances between general channels.
For more information on the distances discussed herein, the reader may consult, e.g.,~\cite{gilchrist2005distance,yuan2017fidelity}.
We defer proofs to \cref{sec:norm-proofs}.

A standard way to measure distances between channels $\calE_1,\calE_2$ operating on a register ``$A$'' is to choose a distance~$d({\cdot},{\cdot})$ on mixed quantum \emph{states} and then consider
\begin{equation} \label{eqn:disty}
    \sup_B \max_{\rho_{AB}} \{d((\calE_1 \otimes \id_B)(\rho_{AB}), (\calE_2 \otimes \id_B)(\rho_{AB}))\},
\end{equation}
where $B$ is an ancilla register and $\rho_{AB}$ is a mixed state on both registers.
When $d$ is trace distance, this leads to the diamond norm:

\begin{definition}\label{def:diam-norm}
    When the state metric in \Cref{eqn:disty} is trace distance, $d(\rho_1, \rho_2) = \half \onenorm{\rho_1 - \rho_2}$, we obtain the \emph{diamond-norm distance}, denoted $\dist_\diamond(\calE_1, \calE_2) = \half \diamnorm{\calE_1 - \calE_2}$.
\end{definition}

The diamond-norm distance is a natural choice for measuring \emph{worst-case} error.
It has direct operational meaning:
if a channel sitting in a quantum circuit is replaced by another channel 
that is $\eps$-different in diamond-norm,
then the output distribution of the algorithm is at most $\eps$-different from the original in total variation distance.
Between unitary channels, the diamond-norm distance 
is equivalent (up to constants) to other familiar metrics on unitary matrices.
We will switch between these metrics throughout the paper depending on which is convenient.
One such metric is operator norm distance up to phase, which we use in \cref{sec:base}.

\begin{definition}\label{def:p-op-norm}
    For unitaries $U, V \in \UU(\dd)$ we define
    \begin{equation}
        \pudist{U}{V} = \min_{\phi \in \uone} \|\phi U - V\|_{\mathrm{op}}.
    \end{equation}
\end{definition}

\noindent
We must take distance up to phase because a unitary channel only specifies a corresponding unitary up to global phase; for example, $U$ and $-U$ perform identical operations on quantum states.
Another equivalent metric is the intrinsic metric on the projective Lie group $\PPUU(\dd) = \UU(\dd)/(\uone\id)$, which we use in \cref{sec:geometry}.

\begin{definition}\label{def:pu-lie-norm}
    We consider the length of a path on the Lie group $\UU(\dd)$ with respect to the operator norm on the tangent space~$\uu(\dd)$.
    For unitaries $U, V \in \UU(\dd)$, we define $\dist(U, V)$ to be the length of a shortest smooth 
    path~$\gamma: [0,1] \to \UU(\dd)$ between them ($\gamma(0) = U, \gamma(1) = V$).
    As an equation,
    \begin{equation}
        \dist(U, V) = \inf_{\text{paths}\,\gamma} \operatorname{len}(\gamma),
        \qquad \text{where}\,\operatorname{len}(\gamma) = \int_0^1  \opnorm{\gamma'(t) \gamma(t)^{-1}} \rd t = \int_0^1  \opnorm{\gamma'(t)} \rd t .
    \end{equation}
    We further define the corresponding metric on the projective group $\PPUU(\dd)$ via this metric.
    \begin{equation}
        \dist(U, \uone V) = \inf_{\phi \in \uone} \dist(U, \phi V)
    \end{equation}
\end{definition}

\noindent An alternative way to define $\dist(U, V)$ is that it is the unitarily invariant metric satisfying $\dist(\id, \diag(e^{\ii \phi_1}, \ldots, e^{\ii \phi_\dd})) = \max_{k \in [\dd]} \abs{\phi_k}$ for $\phi_1,\ldots,\phi_\dd \in [-\pi, \pi)$.%
\footnote{
Taking this definition as a starting point, 
one may find it involved to prove triangle inequality,
which is trivial under the definition by the minimization over paths.
We will later see that these two definitions coincide.
}
That is, $\dist(U, V)$ is the largest angle an eigenvalue of $U^\dagger V$ forms with $1$.

\begin{restatable}[Equivalence of diamond norm, operator norm, intrinsic Lie metrics]{proposition}{normeq} \label{prop:norm-equivalence}
    For unitaries $U, V \in \UU(\dd)$, the metrics $\diamnorm{\calU(U) - \calU(V)}$, $\pudist{U}{V}$, and $\dist(U, \UU(1)V)$ are equivalent up to constants.
    Specifically, the following inequalities hold.
    \begin{gather}
        \pudist{U}{V} \leq \dist(U, \uone V) \leq \tfrac{\pi}{2}\pudist{U}{V} \label{eq:pudist-lie} \\
        \half\diamnorm{\calU(U) - \calU(V)} \leq \pudist{U}{V} \leq \diamnorm{\calU(U) - \calU(V)} \label{eq:diam-pudist}
    \end{gather}
\end{restatable}

Other work on unitary estimation and ``Storage-and-Retrieval'' problem, defined below, considers an accuracy measure that is \emph{not} equivalent: the \emph{entanglement (in)fidelity}.
This is the accuracy measure that naturally arises when one estimates a channel by performing state tomography on its Choi state.
Another way to recover this norm is to consider the channel norm induced by Bures distance (square-root of infidelity), i.e.\ taking $d$ to be Bures distance in \Cref{eqn:disty}, but rather than maximizing over all~$\rho_{AB}$ one simply fixes $\rho_{AB}$ to be the maximally entangled state.\footnote{When $\calE_1$ is unitary, the entanglement infidelity is the same as the so-called  \emph{average gate infidelity} up to the small constant factor of $\dd/(\dd+1)$~\cite{horodecki1999general}. Without this restriction of fixing $\rho_{AB}$, the resulting distance is an alternative distance studied in~\cite{yuan2017fidelity}.  It is termed the \emph{minimum gate fidelity} when~$\calE_1$ is unitary.}
\begin{definition}
    The \emph{entanglement infidelity} between $\dd$-dimensional channels $\calE_1, \calE_2$ is defined to be 
    \begin{equation}
        \ol{F}(\calE_1,\calE_2) = 1 - F(J(\calE_1), J(\calE_2)),
    \end{equation}
    where $F$ is squared quantum state fidelity $F(\rho,\sigma) = (\onenorm{\sqrt \rho \sqrt \sigma})^2$
    and $J(\calE)$ is the Choi state~$(\calE_A \otimes \id_B)(\tfrac{1}{\dd}\sum_{i,j} \ket{ii}\bra{jj})$ of~$\calE$.  
\end{definition}

\noindent
The notation $\ol{F}({\cdot},{\cdot})$ is not standard.
In case $\calE_1 = \calU(U)$ and $\calE_2 = \calU(V)$ for $U,V \in \UU(\dd)$,
we have a simpler formula
\begin{equation} \label{eq:entfid}
    \ol{F}(\calU(U),\calU(V)) = 1 - \abs*{\tfrac{1}{\dd} \Tr(U^\dag V)}^2.
\end{equation}

\begin{remark}
    $\ol{F}(\cdot,\cdot)$ is \emph{not} a metric on~$\PPUU(\dd) = \UU(\dd) / (\uone \id)$, but it is related to one:
    \begin{equation}
        \sqrt{ 1 - \sqrt{1 - \ol F(\calU(U), \calU(V))} }
        = \tfrac{1}{\sqrt{2\dd}}\min_{\phi \in \uone}\fnorm{U - \phi V}
    \end{equation}
    is a metric (which is roughly $\sqrt{\ol F / 2}$ for small~$\ol F$).
    The Frobenius norm on the right-hand side can be thought of as the two-norm distance between the Choi states of $U$ and $V$, treated as pure state vectors, not density matrices, minimized over global, unphysical phase factors.
\end{remark}

An elementary argument shows that estimating a unitary channel to entanglement infidelity $\eps$ only implies estimating it to $\sqrt{\dd\eps}$ error in diamond-norm distance.
The factor of $\sqrt{\dd}$ reflects that entanglement infidelity is an \emph{average-case} accuracy measure, whereas diamond norm is a worst-case accuracy measure.
For example, the entanglement infidelity between 
the multiply controlled NOT gate $C^k X$ (\emph{e.g.}, the usual CNOT if $k=1$ or Toffoli if $k=2$)
and the identity is $O(2^{-k})$ for large~$k$,
but the diamond-norm distance between the two is~$1$.

\begin{restatable}{proposition}{fiddiam} \label{prop:ineq1}
    For all $\dd$ and any unitaries $U,V \in \UU(\dd)$,
    \begin{equation}
        4\, \ol{F}(\calU(U),\calU(V)) \leq \diamnorm{\calU(U) - \calU(V)}^2 \leq 2\dd\cdot\ol{F}(\calU(U),\calU(V)).
    \end{equation}
    The left-hand side is sharp for $\dd$ even and the right-hand side is sharp for all~$\dd \ge 2$:
    \begin{equation}
    \sup_{U \notin \uone V} \frac{\ol{F}(\calU(U),\calU(V))}{\diamnorm{\calU(U) - \calU(V)}^2} = \frac 1 4 \text{ if $\dd$ is even, and}\quad
    \sup_{U \notin \uone V} \frac{\diamnorm{\calU(U) - \calU(V)}^2}{\ol{F}(\calU(U),\calU(V))} = 2\dd.
    \end{equation}
\end{restatable}

\subsection{Unitary estimation and related problems}

We are interested in the following task:
\begin{definition}
    The \emph{estimation task for unitary channels} is the following:  One is given black-box access to a unitary channel by~$Z \in \UU(\dd)$.
    After applying the unitary channel some~$Q$ times, the algorithm should output a classical description\footnote{
        \label{foot:tedium}
        Throughout this paper, for ease of exposition we will be treating these descriptions as if they have infinite precision. To formalize our results when numbers can only be specified up to a machine precision $\macheps$, we consider the sets of unit vectors and unitary matrices that can be generated from $O(\log^3(1/\eps))$-depth quantum circuits.
        This discretizes these spaces into $O(\macheps)$-nets by the Solovay--Kitaev theorem. Then, a \emph{classical description} of a unitary matrix or a unit vector is a list of numbers defining a matrix or vector that need not exactly satisfy the desired constraints, but is $\macheps$-close to an element of the net. When we need to implement a unitary from its classical description, we implement this element. We only care about vectors and matrices up to $\poly(\eps/\dd)$ error, since this is the magnitude of sampling error from state tomography, which is where these classical descriptions come from. So, if we set $\macheps$ to be smaller than this, then the errors from working in finite precision are negligible, and we incur a $\polylog(1/\macheps)$ overhead (which is only $\polylog(\dd/\eps)$) in time complexity. This issue only affects constant factors in our proofs in a very minor way, and we trust the diligent reader to take note of the required changes.
    } of a unitary~$\wh{\bZ} \in \UU(\dd)$.

    The algorithm may be probabilistic and hence $\wh{\bZ}$ should be thought of as a random variable.
    The goal is to achieve 
    \begin{equation}    
        \E[\emeasure(\calU(Z),\calU(\wh{\bZ}))] < \eps.
    \end{equation}
    Here, $\emeasure({\cdot},{\cdot})$ is some accuracy measure between unitary channels.
\end{definition}
Note here we are only using the error measure between two unitary channels.
\begin{remark}\label{rem:confidenceBoosting}
    More generally, one can introduce another parameter~$\eta$ and have the goal be that  $\emeasure(\calU(Z),\calU(\wh\bZ)) \leq \eps$ except with probability at most~$\eta$. 
    As long as $\emeasure({\cdot},{\cdot})$ is a metric, this can be achieved with $O(Q \log(1/\eta))$ applications of~$Z$ as follows:  First, Markov's inequality implies that $\Pr[ \emeasure( \calU(Z), \calU(\wh \bZ) ) > 3\eps ] \le 1/3$.  Then, one can reduce to a general~$\eta > 0$ at the expense of an~$O(\log(1/\eta))$ factor using the ``median trick''; see \Cref{prop:confidence}.
\end{remark}
In order to more carefully quantify the resources needed to solve the unitary estimation task, we generally consider hybrid classical/quantum algorithms.  
These mix classical computation and measurement-outcome-processing with quantum state preparation, applications of~$Z$ and other quantum operations, and measurements. 
Besides the query complexity $Q$, there are some additional resources one wishes to minimize, one of which is the following.
\begin{definition}
    In the task of estimating $Z \in \UU(\dd)$, assume that $\dd = 2^{\nn}$ for some number of qubits~$\nn$.  The \emph{space overhead} of an estimation algorithm is the total number of qubits beyond the minimum,~$\nn$, used in the course of the algorithm.
\end{definition}

There are several tasks closely related to unitary estimation in the literature, going under the names of ``universal programming'', ``unitary cloning'', ``reference frame transmission'', and ``unitary learning''.  For comparison with previous work, it suffices for us to discuss the problem commonly known as ``learning'' a unitary channel. 
However, since ``learning'' and ``estimation'' are often regarded as synonymous in casual speech, we will follow~\cite{sedlak2019optimal} and use the terminology ``Storage-And-Retrieval'', to avoid confusion.

\begin{definition} \label{def:sar}
    The \emph{SAR task for unitary channels} (Storage-And-Retrieval) involves developing two quantum algorithms~\cite{Bisio2009}.  The ``storage'' algorithm~$\calS$ is given black-box access to a unitary $Z \in \UU(\dd)$; after applying the unitary some $Q$~times, $\calS$  should output a possibly mixed quantum state~$\psi$.  The ``retrieval''  algorithm~$\calR$ takes as input~$\psi$ and implements a possibly nonunitary qudit channel $\boldsymbol \calC$, which may be random with some distribution.
    This retrieval may destroy $\psi$.
    The goal is to achieve the following:
    \begin{equation}
        \Pr_{\boldsymbol \calC} [
                \emeasure(\calU(Z), \boldsymbol\calC ) > \eps
            ] \leq \eta.
    \end{equation}
    Note that here $\emeasure({\cdot},{\cdot})$ is measuring the error between a unitary channel and a possibly nonunitary channel.
    The \emph{storage complexity} of the algorithm is the number of qubits used for~$\psi$.  
\end{definition}

\begin{remark}
    The case of minimizing~$\eta$ while insisting on $\eps = 0$ is known as \emph{Probabilistic SAR} and is studied in, e.g.,~\cite{sedlak2019optimal}.
    In this case, the channel~$\boldsymbol\calC$ is necessarily unitary.
    The output of the storage algorithm in~\cite{sedlak2019optimal} is always pure.
    On the other hand, for comparison to prior work, we focus on the case of~$\eta = 0$, which may be termed \emph{Approximate SAR}.  This latter problem is the one commonly known as ``learning'' a unitary.
\end{remark}

We can perform this task with a unitary estimation algorithm by simply taking the estimate's classical description and then synthesizing and applying it.
We provide some notation formalizing this.

\begin{definition}
    If $\bU$ is a $\UU(\dd)$-valued classical random variable, we denote by $\calM(\bU)$ the mixed-unitary channel associated to~$\bU$
    \begin{equation}
        \calM(\bU) = \sum_X \Pr[\bU = X] ~\calU(X).
    \end{equation}   
\end{definition}

Note that if $\bU$ is a unitary-valued random variable,
the notation $\calU(\bU)$ stands for a unitary-channel-valued random variable;
that is, $\calU(\bU)$ is not the same thing as $\calM(\bU)$.
We will now compare unitary estimation and SAR.
An instructive example to keep in mind is the case where an algorithm $\alg$ takes a black box for the identity channel $Z = \id$ and outputs the matrix $\wh{\bZ}$, which is $[\begin{smallmatrix}1 \\ & e^{\ii \eps} \end{smallmatrix}]$ or $[\begin{smallmatrix}1 \\ & e^{-\ii \eps} \end{smallmatrix}]$ with probability $\frac12$.
One can verify that $\E[\diamnorm{\calU(Z) - \calU(\wh{\bZ})}] = \Theta(\eps)$ and $\diamnorm{\calU(Z) - \calM(\wh{\bZ})} = \Theta(\eps^2)$, so by the following remark, on this input $Z = \id$ with diamond-norm distance as the error measure, $\alg$ performs unitary estimation to error $\Theta(\eps)$ and Approximate SAR to error $\Theta(\eps^2)$.

\begin{remark} \label{rem:est-to-SAR}
    A unitary estimation algorithm~$\alg$ may be turned into an SAR algorithm with the same $\eps$~and~$\eta$ parameters: the storage algorithm~$\calS$ is~$\alg$, with~$\psi$ being the mixed state of orthogonal quantum states each of which is classical and encodes an output~$\wh{\bZ}$ of~$\alg$; and the retrieval algorithm~$\calR$ simply synthesizes and applies~$\boldsymbol\calC = \calU(\wh{\bZ})$ from its classical description. Here, the channel $\boldsymbol\calC$ is \emph{unitary}.

    More subtly, we can also turn $\alg$ 
    into an Approximate SAR algorithm as follows.
    We let the storage algorithm the same as before,
    but let the retrieval algorithm implement a mixed unitary $\calC = \calM( \wh \bZ )$
    with probability one.
    This reduction is identical to the previous one, with the difference being that we treat the randomness of $\wh \bZ$ as part of the channel being applied.
    The accuracy of this Approximate SAR is given by
    \begin{equation}    \label{eqn:avgerror}
        \emeasure( \calU(Z), \calM(\wh{\bZ})).
    \end{equation}
    This error measure is different in general from the error measure in the unitary estimation algorithm's promise, $\E_{\wh \bZ} \emeasure(\calU(Z), \calU(\wh \bZ) ) \le \eps$.
\end{remark}

If $\emeasure$ is diamond-norm distance, then $\emeasure(\calU(Z), \calM(\wh \bZ)) \le \E_{\wh \bZ} \emeasure( \calU(Z), \calU(\wh \bZ)) \leq \eps$ by the convexity of the norm, so in this setting Approximate SAR reduces to unitary estimation.
However, $\emeasure(\calU(Z), \calM(\wh \bZ))$ can in fact be made much smaller, to $O(\eps^2)$; see \cref{prop:conv-to-SAR}.
We can think about Approximate SAR as, roughly, the version of unitary estimation where we only need to be right ``in expectation''.
In particular, though the output of the unitary estimation algorithm $\calU(\wh \bZ)$ is explicitly given, the output of the converted Approximate SAR algorithm $\calM(\wh \bZ)$ is not, so we might not know what it is.

\begin{remark}\label{rem:SpecialEntFid}
(Continued from~\cref{rem:est-to-SAR})
If $\emeasure = \ol F$, the entanglement infidelity,
the accuracy of a unitary estimation algorithm is equal to that of its converted Approximate SAR algorithm.
(Note that $\ol F$ is not a metric.)
To see this, recall that the squared quantum state fidelity between a pure state $\psi$
and a mixed state $\rho$ is linear in~$\rho$: $\onenorm{\sqrt \psi \sqrt \rho}^2 = \bra \psi \rho \ket \psi$.
Then, with a maximally entangled state~$\Phi = \ket \Phi \bra \Phi$, we have
\begin{equation}
F(\calU(Z), \calM(\wh \bZ)) = \bra{\Phi} Z^\dag \E_{\wh \bZ} \wh \bZ \Phi \wh{\bZ}^\dag Z \ket \Phi = \E_{\wh \bZ} \bra{\Phi} Z^\dag  \wh \bZ \Phi \wh{\bZ}^\dag Z \ket \Phi = \E_{\wh \bZ} F(\calU(Z), \calU(\wh \bZ) )
\end{equation}
where we omitted subsystem indices and identity tensor factors.
Hence, $\ol F(\calU(Z),\calM(\wh \bZ)) = \E \ol F(\calU(Z),\calU(\wh \bZ) )$.
\end{remark}

One can use this fact about entanglement infidelity to show that the variance of a unitary estimation algorithm bounds the error of its converted Approximate SAR algorithm.
In some sense, this gives a quadratically better version of \cref{rem:est-to-SAR}, since $\eps$-error unitary estimation can give $\eps^2$-error Approximate SAR.

\begin{restatable}[{\cite[Lemma 2]{Yang2020}}]{proposition}{yrcmixing} \label{prop:conv-to-SAR}
    If a unitary estimation algorithm~$\alg$ with an output distribution~$\varphi(\wh \bZ|Z)\rd \wh \bZ$
    is unitarily covariant, i.e., $\varphi(A \wh \bZ B | A Z B ) = \varphi( \wh \bZ | Z )$ for all $A,B \in \UU(\dd)$,            
    then
    \begin{equation}
            \ol F(\calU(Z), \calM(\wh \bZ) ) = \diamnorm{ \calU(Z) - \calM(\wh \bZ) }.
    \end{equation}
    Therefore, by \cref{rem:SpecialEntFid,prop:ineq1},
    \begin{equation}
            \diamnorm{ \calU(Z) - \calM(\wh \bZ) }
            = \ol F(\calU(Z), \calM(\wh \bZ) )
            = \E_{\wh \bZ} \ol F(\calU(Z), \calU(\wh \bZ) )
            \le \tfrac 1 4 \E_{\wh \bZ} [ \diamnorm{\calU(Z) - \calU(\wh \bZ) }^2 ]. \label{mixingIneq}
    \end{equation}
\end{restatable}

\subsection{Prior work}\label{sec:CompPrior}

As far as we are aware, ours is the first work to study unitary estimation with the more stringent diamond-norm distance as its accuracy measure.
We survey the existing literature on related topics, and then compare it to our result.

\paragraph{Standard process tomography.}
Standard quantum process tomography~\cite[Chapter 8.4.2]{nielsen2002quantum} solves the task of general-channel estimation problem by preparing a basis of quantum states, passing them through the channel, and performing state tomography on the results.  One has to take care to analyze how the error bounds for state tomography affect the error bounds for channel estimation; see, e.g.,~\cite{lu2020direct} estimation of minimum gate fidelity for bounds of the form $Q \leq \poly(\dd)/\eps^2$ for general-channel estimation with respect to minimum fidelity.  For the special case of unitary channels one can get improved bounds, as one only needs to work with $O(\dd)$~pure-state estimation tasks. Naively analyzing this strategy yields an $Q \leq O(\dd^3/\eps^2)$ sample upper bound, as an $\eps$ error for each state estimate can compound to a $O(\sqrt{\dd}\eps)$ error in diamond-norm distance for the resulting channel. But as we show in \Cref{sec:base}, with care this method can be used to obtain $Q \leq O(\dd^2/\eps^2)$ query complexity for unitary estimation with respect to diamond-norm distance.  An advantage of this method is that it uses zero space overhead.

Another well-known approach to quantum process tomography is the ancilla-assisted method dating back to Leung~\cite{leung2000towards}. Here one prepares the maximally entangled state in $\dd^2$ dimensions, passes the first half through the channel, and uses state tomography.  Again, if the channel is unitary one can use pure state tomography, and with this approach it is not hard to deduce that $Q \leq O(\dd^2/\delta)$ queries suffices to obtain unitary estimation with respect to \emph{entanglement infidelity~$\delta$}.  Besides using $O(\log \dd)$ space overhead, this approach also only gives an $Q \leq O(\dd^3/\eps^2)$ query complexity bound for diamond-norm distance~$\eps$, via the relations \Cref{prop:norm-equivalence,prop:ineq1}. However, like with standard process tomography, with a tighter analysis, one may be able to give an $Q \leq O(\dd^2/\eps^2)$ query complexity bound.

\paragraph{Unitary estimation.}
The specific task of unitary estimation was perhaps first studied by Ac\'{i}n, Jan\'{e}, and Vidal~\cite{acin2001optimal}, where representation theory was used to determine an optimal algorithmic strategy with respect to entanglement infidelity under the assumption that the unknown~$Z$ is applied in parallel to one half of a bipartite system.  Asymptotic analysis of the error was not given, however.  
Later, Peres and Scudo~\cite{peres2002covariant} gave an alternate method establishing that $Q \leq O(1/\sqrt{\delta})$ queries suffice to obtain entanglement infidelity~$\delta$ in the case $\dd = 2$.
Then Bagan, Baig, and Mu\~{n}oz-Tapia~\cite{bagan2004entanglement} established the same scaling for the method from~\cite{acin2001optimal}; and,  they~\cite{bagan2004quantum} and Chiribella, D'Ariano, Perinotti, and Sacchi~\cite{chiribella2004efficient} did the same for a similar method that didn't require entangled measurements on both parts of the bipartite system.  See also independent work of Hayashi~\cite{hayashi2006parallel}.  
Chiribella, D'Ariano, and Sacchi~\cite{chiribella2005optimal} again showed a more general optimality result for entanglement infidelity, implying that $Q \geq \Omega(1/\sqrt{\delta})$ is a lower bound for this accuracy measure when $\dd = 2$; however, the asymptotic dependence for~$\dd > 2$ was not established.  Later, Kahn~\cite{kahn2007fast} showed that the optimal scaling for general~$\dd$ is of the form $Q = \Theta(f(\dd)/\sqrt{\delta})$ for some function~$f(\dd)$, but was unable to asymptotically analyze it.

Finally, in 2020, Yang, Renner, and Chiribella~\cite{Yang2020} were able to analyze the optimal unitary estimation algorithm for entanglement infidelity and showed that it achieves error~$\delta$ with $Q \leq O(\dd^2/\sqrt{\delta})$ queries.
Moreover, they showed \Cref{prop:conv-to-SAR}, and since their algorithm is unitarily covariant, this implies an Approximate SAR algorithm with $Q \leq O(\dd^2/\sqrt{\eps})$ with respect to diamond-norm error~$\eps$.
We remark that this optimal unitary estimation algorithm from~\cite{Yang2020} applies the unknown unitary~$Z$ in parallel and hence has space overhead on the order of the query complexity~$\Theta(\dd^2\log(\dd)/\sqrt{\eps})$.

\paragraph{Comparison to \cite{Yang2020}.}
Our main result recovers the unitary estimation (with respect to entanglement infidelity) and Approximate SAR (with respect to diamond-norm distance) results of Yang, Renner, and Chiribella~\cite{Yang2020}.
First, by norm conversion \cref{prop:ineq1}, our algorithm with $Q \leq O(\dd^2/\eps)$ achieves the guarantee $\E[~\ol{F}(\calU(\wh{\bZ}),\calU(Z))~] \leq \E[~\diamnorm{\calU(\wh{\bZ})- \calU(Z)}^2~] \le \eps^2$, giving the unitary estimation guarantee.
Second, since any algorithm can be made unitarily covariant, using the same reduction of \cref{prop:conv-to-SAR}, we achieve an Approximate SAR algorithm with $Q \leq O(\dd^2/\sqrt{\eps})$ with respect to diamond-norm error $\eps$.

Neither of these results from \cite{Yang2020} imply our result of unitary estimation to diamond norm error.
Through norm conversion~\cref{prop:ineq1}, a $O(\dd^2/\sqrt{\delta})$-query unitary estimation algorithm with respect to entanglement infidelity $\delta$ implies a $O(\dd^{2.5}/\eps)$-query algorithm with respect to diamond-norm distance $\eps$.
This conversion is tight, even if the unitary estimation algorithm is unitarily covariant.
Also, as discussed in \cref{rem:est-to-SAR}, the output of an approximate SAR algorithm is correct in expectation, but any individual output need not be close in diamond norm.

Finally, our algorithms have space overhead of zero, improving over \cite{Yang2020} and making our algorithm significantly closer to practical.
Note that the main figure of merit in \cite{Yang2020}, \emph{program cost}, denotes the size of the output of the storage algorithm, and is different from the space referred to here, which is the space complexity of the storage algorithm (\cref{def:sar}).

\paragraph{Comparison to \cite{vAcgn22}.}
Van Apeldoorn, Cornelissen, Gily\'{e}n, and Nannicini shows that, given a unitary $Z \in \UU(\dd)$ implementing an unknown state $Z\ket{0} = \ket{z} \in \mathbb{C}^\dd$, one can compute an estimate $\ket{\wh{\bz}}$ which is $\eps$-close in Euclidean norm with $\geq 1-\delta$ probability~\cite[Theorem~23]{vAcgn22}.
Their algorithm uses $O(\frac{\dd}{\eps}\log\frac{\dd}{\delta})$ applications of the controlled black box unitary $\controlled Z$ and its inverse $\controlled Z^\dagger$, along with $O(\dd\log\frac{\dd}{\eps})$ qubits.
This implies an algorithm for unitary estimation, by using this state tomography algorithm for standard process tomography (see \cref{prop:weaker-base,prop:puretomog}) which uses $O(\frac{\dd^2}{\eps}\log\frac{\dd}{\delta})$ applications of $\controlled Z$ and $\controlled Z^\dagger$ and $\Theta(\dd\log\frac{\dd}{\eps})$ space overhead.
This algorithm is based on quantum singular value transformation, so the use of $\controlled Z$ and $\controlled Z^\dagger$ and the nonzero space overhead appear to be inherent limitations of their approach.
Similarly, other standard primitives in quantum algorithms like amplitude estimation, which we might ordinarily reach for when aiming for a quadratic improvement in error, have similar limitations in terms of requiring stronger access to $Z$ or space overhead.
Our algorithm is more direct and so does not lose anything in query complexity or space complexity, and only uses queries to $Z$, and not $\controlled Z$ or $\controlled Z^\dagger$.
The gate complexity of both algorithms is similar.

\paragraph{Comparison of techniques to prior work.}

Existing algorithms for estimating a parametrized class of gates with Heisenberg scaling, like metrology with a GHZ state~\cite{jwdfy08} and robust phase estimation~\cite{Kimmel2015}, proceed by applying the black box many times in parallel or in series (respectively).
Our algorithm uses the same principles (see the warmup in \cref{fig:warmup}), but as discussed in the next section, naive generalizations of these approaches cannot work.
We introduce a novel ``shift to identity'' step to avoid areas in the space of unitary channels where applying the black box in series fails; this adaptivity allows us to extend robust phase estimation, which can estimate a parameter on the complex unit circle, to work for elements in a much trickier space.

The process tomography protocol most similar to the algorithm presented in our work is that of \emph{gate set tomography}~\cite{Nielsen2020}, which can achieve Heisenberg scaling, $\Theta(1/\eps)$, in regimes of practical interest by interleaving long sequences of gates.
This is qualitatively similar to what our algorithm does, though we are not aware of work giving theoretical bounds on such protocols.

\paragraph{Prior work on lower bounds.}
Lower bounds on channel estimation and discrimination are just as well-studied as upper bounds.
Some prior work has tried to characterize the optimal strategy for various problems, such as showing they are sequential or parallel~\cite{Duan2006,Bavaresco2021}.
As mentioned previously, though these results describe optimal strategies for these problems, 
it is unclear if they imply lower bounds better than $\Omega(\dd^2 + \frac{1}{\eps})$.
Up to log factors, this lower bound is an immediate consequence of parameter counting (e.g.\ \cite[Chapter 8.4.2]{nielsen2002quantum}, \cite{bkd14}, which one can make rigorous with a Holevo bound~\cite{Cleve1997}) and the optimality of Heisenberg scaling.
However, these two arguments do not combine naturally to get the $\Omega(\dd^2/\eps)$ bound we would expect.

To our knowledge, we give the first \emph{jointly optimal} lower bound, in the sense that 
for any decaying function $\dd \mapsto \eps = \varphi(\dd)$
our lower bound $\Omega(\dd^2 / \varphi(\dd))$ 
matches the query complexity of our algorithm up to a universal constant.
Our lower bound combines prior work on unitary channel discrimination~\cite{Bavaresco2021} with a reduction technique from quantum query complexity~\cite{cgmsy09,bccks17}.
Though quantum query complexity focuses on diagonal unitaries, we observe that roughly the same technique can be applied to this general unitary estimation setting with some modification.

\subsection{Techniques} \label{sec:techniques}

\paragraph{An algorithm with suboptimal $\eps$ scaling.} A standard idea in process tomography is to estimate $Z$ by doing state tomography on $Z|0\rangle$, $Z|1\rangle$, and so on, and then somehow collate these estimates into a full estimate of $Z$.
Since each state tomography requires $O(\dd/\eps^2)$ samples, this whole procedure can be done with $O(\dd^2/\eps^2)$ applications of $Z$ (\cref{thm:base}).
There are two minor issues to address when formalizing this.
First, an $\eps$ error in each state tomography could cascade to a $\eps\sqrt{\dd}$ error in the final estimate of $Z$.
We address this by ensuring that the state tomography algorithm produces Haar-random error, so that with high probability the errors do not compound.
Second, we can only estimate each column up to a phase, so we need to additionally deduce the relative phases between columns.
We address this by showing that learning the columns up to phase of both $Z$ and $Z F$, $F$ being the discrete Fourier transform, suffices to deduce these relative phases (\cref{prop:weaker-base}).

\paragraph{Achieving the optimal Heisenberg scaling.}
We wish to improve the $O(1/\eps^2)$ dependence of process tomography to $O(1/\eps)$.
Our strategy will be to use process tomography as a subroutine, but boost the error some other way.
Specifically, we give a bootstrapping procedure (\cref{algo:theboot}) that estimates a unitary to $\eps$ error by making calls to a ``base'' algorithm that can only output unitary estimates to error $\frac{1}{200}$, 
paying an $1/\eps$ factor of overhead in the query complexity (\cref{thm:error-boosting}).
Using the $O(\dd^2/\eps_0^2)$ unitary estimation algorithm with $\eps_0 = \frac{1}{200}$, as the base, 
this gives the desired $O(\dd^2/\eps)$ query complexity.

\begin{figure}[ht]
    \centering
    \includegraphics{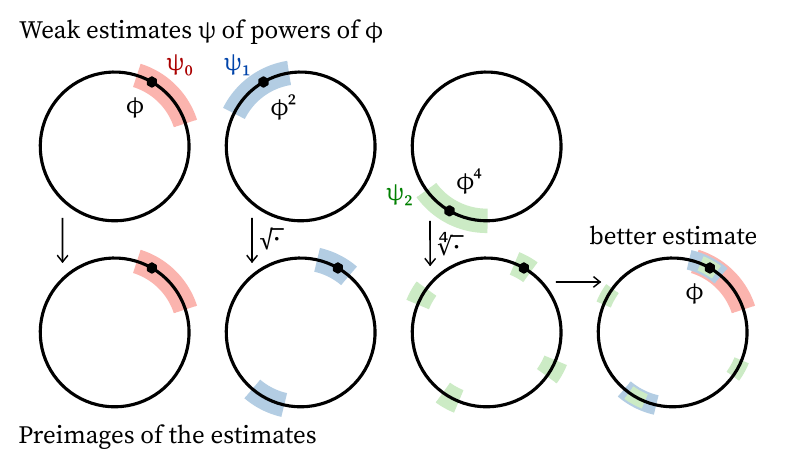}
    \caption{The warmup example: we wish to learn $Z = \big(\begin{smallmatrix} 1 \\ & \phi \end{smallmatrix}\big)$, where we do not know the phase $\phi$ on the complex unit circle.
    Our strategy is to get constant-error estimates to $Z^{2^k}$, which correspond to phases $\psi_k$ that specify a constant-sized interval around $\phi^{2^k}$ for each $k$.
    Although each interval only pins down the corresponding power of $\phi$ to constant error, each $\psi_k$ can be thought of as specifying the $k$th ``bit'' of information of $\phi$.
    So, when taken together, these error intervals can be collated to get an estimate of $\phi$ to error $O(\eps)$ for the cost of computing constant-error estimates of powers of $Z$ up to $Z^{1/\eps}$.}
    \label{fig:warmup}
\end{figure}

A good first try is to get constant-error estimates for $Z^{2^j}$ for $j$ going from 1 to $\lceil \log_2(1/\eps)\rceil$, with the hope that the estimate of $Z^{2^j}$ will refine the estimate of $Z$ to $2^{-j}$ error.
Taking powers of an unknown gate with exponentially increasing degree is an existing method in quantum metrology for achieving the Heisenberg limit.
In particular, we can view this as a non-coherent version of phase estimation~\cite{Kimmel2015}, where the application of $Z^{2^j}$ to an eigenvector of $Z$ extracts the $j$th bit of its corresponding eigenvalue.
This argument can be shown to work in simple cases, such as when $Z$ is a rotation in two dimensions (\cref{fig:warmup}), 
but if all the eigenvalues of $Z$ are $\pm 1$,
then any power is either~$\id$ or~$Z$ so it is impossible 
to gain any better information about the eigenvectors of~$Z$ from its powers.

Nonetheless, we find that 
if the powers of $Z$ are always close to the identity for all sufficiently high exponents (and hence no eigenvalue is close to $-1$), then this idea works out.
Specifically, we show in \cref{lem:sqrt} that if unitaries $U$ and $V$ are $\alpha$-close in diamond-norm distance, 
then $U^{1/p}$ and $V^{1/p}$ are $50\alpha/p$-close, provided that $U$ and $V$ are $0.01$-close to the identity.
We can use this lemma to bootstrap a constant-error estimate to an $\eps$-error one, 
provided we always apply the base algorithm to matrix powers that are close to the identity.
Thus, we always bring the unknown unitary close to the identity: 
instead of running the base unitary estimation algorithm on $Z^{2^j}$, 
we run it on $(Z V_j^\dagger)^{2^j}$, where $V_j$ is our best estimate to $Z$ so far in the algorithm.
This recenters the unitary at the identity so that we can power it up even further.
We note here that, upon formulating the right technical lemma to use (\cref{lem:sqrt}), 
the analysis of the resulting algorithm (\cref{algo:theboot}) becomes surprisingly simple.

\paragraph{Lower bound.}
To achieve our lower bound of $\Omega(\dd^2/\eps)$,
we consider a hard instance of the unitary estimation problem,
in which one is asked to identify one of $\exp(\Omega(\dd^2))$ candidate unitary channels
that are $\eps$-apart from one another 
and are all $O(\eps)$-close to the identity
in the diamond norm.
Specifically, we choose the ensemble to be $\eps$th-power of a net of reflections.

Simple arguments give a lower bound of $\Omega(\dd^2 + \frac{1}{\eps})$, with $\dd^2$ being a generic lower bound for identifying one of $\exp(\Omega(\dd^2))$ candidate unitary channels~\cite{Bavaresco2021} and $\frac{1}{\eps}$ being the number of applications of a channel necessary to discriminate two channels that are $\eps$-close in diamond-norm distance.
To improve this, we wish to argue that,
since a unitary~$Z$ taken from the hard instance applies an $\eps$th power of a reflection,
the task of distinguishing $Z$ is a factor of $1/\eps$ harder than the task of distinguishing the net of reflections.
Intuitively, this holds because $Z$ is essentially the identity for all but an $\eps$ fraction of the time it is applied.

This heuristic picture can be made precise by showing a reduction from ``fractional-query'' algorithms to ``discrete-query'' algorithms~\cite{cgmsy09,bccks17}.
This reduction converts a circuit that calls $Z$, the $\eps$th power of a reflection, into one that calls that reflection conditioned on some ancilla.
So, hardness of a problem using queries to the reflection lifts to hardness of the problem using queries to $Z$.
Translating to our setting, directly applying this reduction shows that a $\Omega(d^2)$ to distinguishing a net of reflections translates to a $\Omega(\frac{\dd^2}{\eps}\frac{\log\log(1/\eps)}{\log(1/\eps)})$ lower bound for the hard instance.
This reduction is tight, though, so removing the $\frac{\log\log(1/\eps)}{\log(1/\eps)}$ term requires additional insight.

We consider again the reduction of \cite{cgmsy09,bccks17} converting a fractional-query circuit to a discrete-query circuit.
This reduction adds one ancilla qubit per oracle call to~$Z$
and the joint state of all the ancillas 
is a superposition of some bitstrings.
The average weight of the bitstring is an effective number of the oracle calls 
to the reflection,
which is smaller than the naive query complexity for~$Z$ by a factor of~$\eps \ll 1$.
Without much change in the output distribution of the overall algorithm,
one can modify the algorithm to monitor the weight of the ancilla bitstring
and actually reduce the number of oracle calls by the factor of~$\eps$.
This reduction however comes at a price:
the modified algorithm now requires postselection
whose success probability is exponentially small in the query number.
Addressing this small success probability is where one picks up the sub-logarithmic term.
However, for this particular lower bound we need not address it: an exponentially small success probability is still hard to achieve for unitary discrimination.
Bavaresco, Murao, and Quintino~\cite{Bavaresco2021} gives a useful relation between
the query complexity of a unitary discrimination problem
and success probability,
and we observe that small success probability is still meaningful
as long as the algorithm is appreciably better than merely guessing the answer.
This gives the optimal $\Omega(\dd^2/\eps)$ lower bound.

\subsection{Discussion}

\paragraph{Improving gate complexity.}
We give an algorithm for unitary channel estimation that is query-efficient and space-efficient.
However, we still need to run quantum circuits with $\poly(\dd, 1/\eps)$ depth, in addition to the oracle calls.
If this could be improved to $\poly(\log(\dd), 1/\eps)$, it would make this algorithm significantly more practical.
However, the depth is bottlenecked by the ``shifting to identity'' step, which requires a unitary synthesis and so is high depth.
As discussed in \cref{sec:techniques}, this step appears to be necessary to avoid the hard case of unitaries with $-1$ eigenvalues.

A simpler question is whether the $O(\dd^2/\eps^2)$ algorithm can be made gate-efficient.
In view of this goal, we give a gate-efficient version of pure state tomography with optimal sample complexity in \cref{app:gate-state}; prior gate-efficient algorithms lose log factors in the sample complexity.
This should be able to give a $O(\dd^2/\eps^2)$ query complexity algorithm with the desired gate complexity and optimal space complexity, though we do not prove this.

\paragraph{Estimating the eigenvalues of a unitary channel.}
We show in \cref{sec:eigenvalues} that shifting to the identity is not necessary when we merely want to learn the eigenvalues of the unitary without the eigenvectors.
In this Appendix, we describe how to use a (gate-efficient) control-free version of phase estimation to achieve this result (though notably, the estimates we achieve do not distinguish between eigenvalues of different multiplicity).
This algorithm appears to be folklore.
If one can reduce unitary estimation to $\poly(\dd)$ instances of eigenvalue estimation, 
then this result would give gate-efficient unitary channel tomography with query complexity $\poly(\dd)/\eps$.

\paragraph{Generalizing beyond unitary channels.}
Finally, we note that our results do not extend to general channels.
For example, the channel that destroys the input and outputs the outcome of a biased coin requires $\Omega(1/\eps^2)$ queries to estimate.
However, we leave open the question of analyzing this algorithm's tolerance to noise, and the related question of whether $o(1/\eps^2)$ process tomography is possible for ``close-to-unitary'' channels.

\section{\texorpdfstring{Base tomography using $O(\dd^2/\eps^2)$ applications}{Base tomography using O(d²/ε²) applications}} \label{sec:base}

In this section we show a base tomography algorithm that has a quadratic dependence on the desired precision.
We will use the ``operator norm up to phase'' distance for this section, defined in \cref{def:p-op-norm}, which is equivalent to diamond norm up to constants.

\begin{theorem} \label{thm:base}
    There is a tomography algorithm that, given access to an unknown unitary $Z \in \UU(\dd)$, as well as parameters $\eps, \eta > 0$, applies $Z$ at most $O(\dd^2/\eps^2) \cdot \log(1/\eta)$ times and outputs an estimate~$\bV$ satisfying $\pudist{Z}{\bV} \leq \eps$ except with probability at most~$\eta$.
    The algorithm uses only a $\dd$-dimensional Hilbert space (in particular, only $\nn$~qubits when $\dd = 2^{\nn}$). 
\end{theorem}
In fact, by standard methods (explicitly shown in \Cref{prop:confidence} below) it suffices to prove this theorem for a fixed confidence value less than~$\frac12$, say $\eta = \frac{1}{3}$.

Our algorithm will essentially learn the unknown~$Z$ column by column, using the below pure state tomography result, \Cref{prop:puretomog}.  This result was essentially previously known, but we will take care of a few minor details in \Cref{sec:ancillary}.
\begin{proposition}    \label{prop:puretomog}
    There is a pure state tomography algorithm with the following behavior.  Given access to copies of a pure state $\ket{z} \in \CC^{\dd}$, 
    it sequentially and nonadaptively makes von Neumann measurements
    on~$O(\dd/\eps_0)$ copies of $\ket{z}$ (using only $\dd$-dimensional Hilbert space).
    Then, after classically collating and processing the measurement outcomes, it outputs (a classical description of) an estimate pure state
    \begin{equation}    \label{eqn:puretomog}
        \ket{\bu} = \bphi \sqrt{1-\beps}\ket{z} + \sqrt{\beps}\ket{\bw}
    \end{equation}
    such that: (i)~$\bphi$ is a complex phase; (ii)~the infidelity~$\beps$ is at most $\eps_0$ except with probability at most~$\exp(-5\dd) \leq \frac{1}{100\dd}$; (iii)~the vector~$\ket{\bw}$ is distributed Haar-randomly\footnote{Perfect Haar-randomness is only possible by making idealized assumptions about the algorithm's hardware; such technical issues of algorithmic complexity are deferred to \cref{foot:tedium}.} among all states orthogonal to~$\ket{z}$.
\end{proposition}
This column-by-column technique has the minor downside that each column estimate can be off by a different complex phase~$\bphi$.
However, there are a few simple ways to work around this flaw; in particular, the following result (proven in \Cref{sec:ancillary} below) gives a completely black-box method:

\begin{proposition} \label{prop:weaker-base}
    Let $\tomalg$ be a tomography algorithm as in \Cref{thm:base}, except with the following weaker guarantee about~$\bV$:
    \begin{equation}    \label{eqn:weaker-base}
        \exists \text{ a diagonal unitary~$\Phi$ such that } \opnorm{Z \Phi - \bV} \leq \eps \leq \tfrac{1}{8}.
    \end{equation}
    Then there is a tomography algorithm (using~$\tomalg$ twice) that achieves \Cref{thm:base}, with~$25\eps$ and $2\eta$ in place of $\eps$ and~$\eta$, and $O(\dd^3)$ additional classical time complexity.
\end{proposition}

Putting all of the above together, we can establish our base tomography result:
\begin{proof}[Proof of \Cref{thm:base}]
    From the preceding discussion, it suffices to obtain a unitary tomography routine achieving \Cref{eqn:weaker-base} with failure probability at most~$\eta = \frac{1}{6}$.
    We apply the pure state tomography routine from \Cref{prop:puretomog}, with $\eps_0 = \Omega(\eps^2)$ to be chosen later, on each of $\ket{z_1}, \dots, \ket{z_{\dd}}$, where $|z_j\rangle = Z|j\rangle$.
    This indeed uses $Z$ at most $O(\dd^2/\eps^2)$ times, within a Hilbert space of dimension only~$\dd$, and produces estimates $\ket{\bu_1}, \dots, \ket{\bu_{\dd}}$.
    Let $\bU \in \CC^{\dd \times \dd}$ denote the (possibly nonunitary) matrix with the $\ket{\bu_j}$'s as columns; we will show that 
    \begin{equation}    \label{eqn:half-base}
        \exists \text{ diagonal unitary~$\Phi$ such that } \opnorm{Z \Phi - \bU} \leq \eps/2,
    \end{equation}
    except with probability at most~$\frac{1}{50}$.
    Then if we output any unitary~$\bV$ satisfying $\opnorm{\bU - \bV} \leq \eps/2$ 
    (for example, $\bV = \bX\bY^\dagger$, where $\bU = \bX\bSigma \bY^\dagger$ is a singular value decomposition),
    then $\bV$ will satisfy \Cref{eqn:weaker-base}, as desired.
    
    To establish \Cref{eqn:half-base}, write
    \begin{equation}   
        \ket{\bu_j} = \bphi_j \sqrt{1-\beps_j}\ket{z_j} + \sqrt{\beps_j}\ket{\bw_j}
    \end{equation}
    for $j \in [d]$ as in \Cref{eqn:puretomog}.
    Then with $\Phi = \diag(\bphi_1, \dots, \bphi_d)$ and $\bW$ the matrix with $\ket{\bw_j}$'s as columns, we have
    \begin{equation}
        \bU - Z \Phi = Z \Phi \bD + \bW \bE, \quad \text{where} \quad \bD \coloneqq \diag\parens[\Big]{\parens[\big]{\sqrt{1-\beps_j} - 1}_j}, \ \bE \coloneqq \diag(\sqrt{\beps}). 
    \end{equation}
    We have 
    \begin{equation}
        \opnorm{\bD} \leq \opnorm{\bE} = \max\{\sqrt{\beps_j}\} \leq \sqrt{\eps_0}
    \end{equation}
    except with probability at most~$\frac{1}{100}$, by a union bound.
    We will also shortly show:
    \begin{equation}    \label{eqn:claimW}
        \textbf{Claim:} \qquad \opnorm{\bW} \leq C \quad \text{except with probability at most~$\tfrac{2}{100}$.}
    \end{equation} 
    (Here $C$ is a universal constant.)
    Combining the above, we conclude
    \begin{equation}
        \opnorm{ {\bU} - Z \Phi} \leq 1 \cdot 1 \cdot \sqrt{\eps_0} + C \cdot \sqrt{\eps_0}
    \end{equation}
    except with probability at most~$\frac{1}{100} + \frac{2}{100} < \frac{1}{6}$, and this norm bound is at most~$\eps/2$ (as needed for \Cref{eqn:half-base}) provided we take the constant in $\eps_0 = \Omega(\eps^2)$ small enough.

    It remains to prove the claim from \Cref{eqn:claimW}.
    Recall that~$\bW$ has independent unit columns $\ket{\bw_j}$, with $\ket{\bw_j}$ Haar-random orthogonal to~$\ket{z_j}$.  
    Introduce i.i.d.\ real random variables $\bdelta_1, \dots, \bdelta_{\dd}$, where~$\bdelta_j$ is distributed as $\abs{\braket{1|\bx}}^2$ for $\ket{\bx}$ a Haar-random unit vector in~$\CC^{\dd}$.
    If we further introduce i.i.d.\ uniformly random complex phases $\bpsi_1, \dots, \bpsi_{\dd}$, then the unit vectors
    \begin{equation}    \label{eqn:Y-from-W}
        \ket{\by_j} \coloneqq \sqrt{\bdelta_j}\bpsi_j \ket{z_j} + \sqrt{1-\bdelta_j}\ket{\bw_j}
    \end{equation}
    are in fact Haar-random and independent.
    Letting~$\bY$ denote the matrix with the $\ket{\by_j}$'s as columns, it is a standard fact in random matrix theory\footnote{For example, if~$\bY$'s columns were independent Haar-random unit vectors in~$\RR^{\dd}$ (as opposed to~$\CC^{\dd}$) then \cite[Theorems~3.4.6,~4.6.1]{vershynin18} would directly yield that $\opnorm{\bY} \leq K$ except with probability at most~$2\exp(-cK^2\dd)$, for some constant $c > 0$. The generalization to the complex case is very minor.} that, for some universal constant~$K$, we have \begin{equation}
        \opnorm{\bY} \leq K  \text{ except with probability at most $\tfrac{1}{100}$}.
    \end{equation}
    We can now rewrite \Cref{eqn:Y-from-W} as 
    \begin{equation}
        \bY = Z \bDelta_0 + \bW \bDelta_1, \quad \text{where} \quad 
        \bDelta_0 \coloneqq \diag\parens[\Big]{\parens[\big]{\sqrt{\bdelta_j} \bpsi_j}_j}, \quad 
        \bDelta_1 \coloneqq \diag\parens[\Big]{\parens[\big]{\sqrt{1-\bdelta_j}}_j}.
    \end{equation}
    from which we can conclude
    \begin{equation}
        \opnorm{\bW} \leq (\opnorm{\bY} + \opnorm{Z \bDelta_0})\cdot \opnorm{\bDelta_1^{-1}} \leq (K+1) \cdot (1-{\max}_j\{\bdelta_j\})^{-1/2}
    \end{equation}
    except with probability at most $\frac{1}{100}$.
    Thus to complete the proof of the claim in \Cref{eqn:claimW}, it suffices to show 
    \begin{equation} \label{eqn:finalmax}
        {\max}_j\{\bdelta_j\} \leq 1 - 1/K' \quad \text{except with probability at most $\tfrac{1}{100}$}
    \end{equation} 
    for some constant~$K'$.
    Note that for each constant value of~$\dd$, the random variable~$\bdelta_j$ has a continuous probability density on~$[0,1]$; from this observation, it's easy to deduce that it suffices to prove \Cref{eqn:finalmax} under the assumption $\dd \geq \dd_0$ for some constant~$\dd_0$.
    But this is easy: $\bdelta_j^2$~has mean~$1/\dd$ and is sub-exponentially distributed with parameter~$\Theta(1/\dd)$ \cite[Proposition~2.7.1, Lemma~2.7.6, Theorem~3.4.6]{vershynin18}; hence $\Pr[\bdelta_j^2 \geq 1/2] \leq \exp(-\Omega(\dd))$ and---with a union bound---this is more than sufficient for \Cref{eqn:finalmax}, once~$\dd$ is sufficiently large.
\end{proof}

\subsection{Ancillary results for base tomography} \label{sec:ancillary}

We begin by giving a proof of the pure state tomography result we needed:
\begin{proof}[Proof of \Cref{prop:puretomog}]
    This result was essentially proven in~\cite{CL14,KUENG201788}.
    To be precise, we will refer to the analysis from~\cite[Theorem~2]{Guta2020}.
    The result therein is exactly what we need, except for the following distinctions:
    \begin{itemize}
        \item Rather than making von Neumann measurements, \cite{Guta2020} refer to performing the ``uniform POVM'' on each copy of~$\ket{z}$; this is the continuous POVM with elements labeled by unit vectors $\ket{v} \in \CC^{\dd}$, in which the $\ket{v}$-element has density $\dd \cdot \ket{v}\!\bra{v}$ with respect to Haar measure on~$\ket{v}$.
        However, this is mathematically equivalent to first using classical randomness to choose a Haar-random~$\bV \sim \UU(\dd)$, and then projectively measuring in the basis of~$\bV$'s columns.
        \item The classical post-processing algorithms in \cite{KUENG201788,Guta2020} do not necessarily output a pure (rank-one) hypothesis; they output a possibly mixed state~$\brho$, counting it as a success (in the case of \cite{Guta2020}) when $\onenorm{|z\rangle\langle z| - \brho} \leq \sqrt{\eps_0}$.
        This also means they do not explicitly confirm condition~(iii) in \Cref{prop:puretomog}, concerning the Haar-randomness of~$\ket{\bw}$.
    \end{itemize}
        
    But inspection of the actual algorithm in \cite{Guta2020} shows that this second issue is easily fixed. The algorithm first forms $\bL = (\dd+1)\mathrm{avg}\{\ket{\bv_j}\!\bra{\bv_j}\} - \id$, where the $\ket{\bv_j}$'s are the measurement outcomes.  
    With $O(\dd/\eps_0)$ measurements, this matrix is shown to satisfy $\opnorm{|z\rangle\langle z| - \bL} \leq \frac14\eps_0$ except with probability at most~$\exp(-5\dd)$. 
    The authors of \cite{Guta2020} then ``round''~$\bL$ to a quantum state~$\brho$ by first diagonalizing it as $\bW \diag(\blambda) \bW^\dagger$ for $\bW \in \UU(\dd)$ and $\blambda \in \CC^{\dd}$, and then taking $\brho = \bW \diag(\blambda') \bW^\dagger$, where $\blambda'$ is the nearest probability vector to~$\blambda$.  
    With this adjustment, they show that $\onenorm{\ket{u}\!\bra{u} - \brho} \leq \sqrt{\eps_0}$ as needed. 

    Note that this last inequality implies that the \emph{closest} rank-$1$ matrix~$\bM$ to~$\brho$ must satisfy \mbox{$\onenorm{\bM - \brho} \leq \sqrt{\eps_0}$}.  
    On the other hand, it is well known that~$\bM$ is simply given by $\bW \diag(\blambda'') \bW^\dagger$, where~$\blambda''$ is formed from~$\blambda'$ by zeroing out all entries except the largest.
    It follows that the zeroed-out entries sum to at most~$\sqrt{\eps_0}$, and hence $\Tr(\bM) \geq 1 - \sqrt{\eps_0}$.
    If we now form $\widetilde{\bM} = \frac{\bM}{\Tr(\bM)}$, then $\widetilde{\bM}$ is a rank-$1$ state with $\onenorm{\widetilde{\bM} - \bM} \leq O(\sqrt{\eps_0})$ and hence $\onenorm{\widetilde{\bM} - \brho} \leq O(\sqrt{\eps_0})$ and $\onenorm{\ket{u}\!\bra{u} - \widetilde{\bM} } \leq O(\sqrt{\eps_0})$.
    Thus we may express $\widetilde{\bM} = \ket{\wh{\bu}}\!\bra{\wh{\bu}}$ and output~$\ket{\wh{\bu}}$ (after slightly adjusting the constant on~$\eps_0$).

    It only remains to observe that the process of forming~$\bL, \brho, \bM, \widetilde{\bM}, \ket{\wh{\bu}}$ is completely symmetric with respect to the subspace orthogonal to~$\ket{u}$, and hence the ``error vector''~$\ket{\bw}$ is indeed distributed Haar-randomly.
\end{proof}

Next, we give the algorithm for \Cref{prop:weaker-base}, which fixes the column phases for our unitary tomography algorithm:
\begin{proof}[Proof of \Cref{prop:weaker-base}]
Let $F \in \UU(\dd)$ denote the discrete Fourier transform, so $\braket{a|F|b} = (1/\sqrt{\dd})\exp(-2\pi \ii a b/\dd)$.
(The only property we use about~$F$ is that each entry has the same magnitude; if $\dd = 2^\nn$, one may use the Hadamard transform $H^{\otimes \nn}$ instead.)
Given~$Z$, our algorithm applies $\tomalg$ once to $Z$ and once to~$Z F^\dagger$; call the results $\bV, \bG \in \UU(\dd)$, respectively.
Except with probability at most~$2\eta$, we get that there exist diagonal unitaries $\bPhi_\bV$~and~$\bPhi_\bG$ such that
\begin{equation}    \label{eqn:eps}
    \opnorm{Z\bPhi_\bV - \bV},\ \opnorm{ZF^\dagger\bPhi_\bG - \bG} \leq \eps.
\end{equation}
Given this, we conclude that $\bG^\dagger \bV$ is close to $F$, up to phases on rows and columns.
\begin{align}
    \opnorm{\bG^\dagger\bV - \bPhi_\bG^\dagger F \bPhi_\bV}
    &\leq \opnorm{(\bG - ZF^\dagger \bPhi_\bG)^\dagger\bV} + \opnorm{(ZF^\dagger \bPhi_\bG)^\dagger(\bV - Z\bPhi_\bV)} \nonumber\\
    &= \opnorm{\bG - ZF^\dagger \bPhi_\bG} + \opnorm{\bV - Z\bPhi_\bV}
    \leq 2\eps \label{eq:dephasing}
\end{align}
We can then deduce $\bPhi_\bV$ by computing $\bG^\dagger \bV$ in $O(\dd^3)$ time, and then essentially reading off the relative column phases.
First, notice that by pigeonhole principle applied to \cref{eq:dephasing},
\begin{align}
    \text{for every }b \in [\dd],\,\abs[\Big]{\braket{a|\bG^\dagger \bV|b} - \braket{a|\bPhi_\bG^\dagger F \bPhi_\bV|b}} \leq 4\eps/\sqrt{\dd} \text{ for at least a } \tfrac{3}{4} \text{ fraction of } a \in [\dd].
    \label{eq:pigeon}
\end{align}
Let $C(a,b)$ denote the inequality in \cref{eq:pigeon}.
Since $\braket{a|\bPhi_\bG^\dagger F \bPhi_\bV|b} = \overline{\bphi}_\bG^{(a)}\bphi_{\bV}^{(b)}\braket{a|F|b}$ has magnitude $\frac{1}{\sqrt{\dd}}$, we can define $P \in \mathbb{C}^{\dd \times \dd}$ to be $\braket{a|P|b} = \braket{a|\bG^\dagger \bV|b}/\braket{a|F|b}$, and it follows that
\begin{align}
    C(a,b)
    &\implies \abs{\braket{a|P|b} - \overline{\bphi}_\bG^{(a)}\bphi_{\bV}^{(b)}} \leq 4\eps; \\
    C(a,1) \text{ and } C(a, b)
    &\implies \abs[\Big]{\frac{\braket{a|P|b}}{\braket{a|P|1}} - \frac{\bphi_\bV^{(b)}}{\bphi_\bV^{(1)}}} \leq \frac{2\cdot 4\eps}{1-4\eps} \leq 16\eps. \label{eq:relative-phase}
\end{align}
By \cref{eq:pigeon}, $C(a,1)$ and $C(a,b)$ hold for at least half of $a$'s.
Let $\bpsi_b$ denote the coordinate-wise (real and imaginary) median of $\{\braket{a|P|b}/\braket{a|P|1}\}_{a \in [\dd]}$.
Then by \cref{eq:relative-phase},
\begin{align}
    \abs[\Big]{\bpsi_b - \bphi_\bV^{(b)}/\bphi_\bV^{(1)}}
    &\leq \sqrt{\Real(\bpsi_b - \bphi_\bV^{(b)}/\bphi_\bV^{(1)})^2 + \Imag(\bpsi_b - \bphi_\bV^{(b)}/\bphi_\bV^{(1)})^2}
    \leq \sqrt{(16\eps)^2 + (16\eps)^2}
    \leq 24\eps.
\end{align}
Consequently, $\abs{\bpsi_b/\abs{\bpsi_b} - \bphi_\bV^{(b)}/\bphi_\bV^{(1)}} \leq 48\eps$.
Let $\bPsi$ be the (unitary) matrix with $\bpsi_b/\abs{\bpsi_b}$ on the diagonal.
Then $\opnorm{\bphi_\bV^{(1)}\bPsi - \bPhi_\bV} \leq 24\eps$.
Combining this with $\opnorm{Z\bPhi_\bV - \bV} \leq \eps$ from \Cref{eqn:eps} easily yields
\begin{equation}
    \pudist{Z}{\bV\bPsi^\dagger}
    \leq \opnorm{\bphi_\bV^{(1)}Z - \bV\bPsi^\dagger}
    \leq \opnorm{\bphi_\bV^{(1)}Z - Z\bPhi_\bV\bPsi^\dagger} + \opnorm{Z\bPhi_\bV\bPsi^\dagger - \bV\bPsi^\dagger}
    \leq 25\eps,
\end{equation}
and so our algorithm may output $\bV \bPsi^\dag$.
\end{proof}

Finally, we give the (completely standard) trick for boosting confidence:
\begin{proposition} \label{prop:confidence}
    Let $\mathcal{L}$ be a learning algorithm for objects in a metric space with efficiently computable distance $\mathrm{distance}({\cdot},{\cdot})$.
    Assume that on input~$Z$, the algorithm $\mathcal{L}$ outputs $\bV$ satisfying $\mathrm{distance}(Z,\bV) \leq \eps$ except with probability at most~$.49$.
    Then there is an algorithm~$\mathcal{L}'$ that takes an additional input~$\eta > 0$, uses $\mathcal{L}$ at most $O(\log(1/\eta))$ times, and guarantees $\mathrm{distance}(Z,\bV) \leq 3\eps$ except with probability at most~$\eta$.
\end{proposition}
\begin{proof}
    For $T = O(\log(1/\eta))$ times, run~$\mathcal{L}$ independently, obtaining $\bV_1, \dots, \bV_T$.  
    By a standard Chernoff bound, except with probability at most~$\eta$, there are at least $.505T$ ``good'' estimates, where we say $\bV_j$ is ``good'' if $\dist(Z,\bV_j) \leq \eps$.  
    By the triangle inequality, every good estimate is also ``central'', where we say estimate $\bV_{j_0}$ is ``central'' if it has the following property:  $\dist(\bV_{j_0}, \bV_k) \leq 2\eps$ for at least $.505T$ estimates $\bV_k$.
    Let $\mathcal{L}$ now select and output any central estimate~$\bV_{j^*}$; the method of brute-force checking central-ness has time complexity $O(T^2)$ times the cost of computing a distance, which is $O(\dd^3)$ for $\dd \times \dd$ matrices.
    Since $.505T + .505T > T$, the Pigeonhole Principle implies that at least one of the estimates $\bV_k$ for which $\dist(\bV_{j^*}, \bV_k) \leq 2\eps$ is also good.  
    Thus the triangle inequality implies $\dist(\bV_{j^*}, Z) \leq 3\eps$, as desired.
\end{proof}

\section{\texorpdfstring{Bootstrap of precision to $O(\dd^2/\eps)$ applications}{Bootstrap of precision to O(d²/ε) applications}}

\subsection{Key lemma: Geometry of unitary groups} \label{sec:geometry}

The purpose of this subsection is to recall 
some well-known metric notions
on unitary groups and projective unitary groups; 
see~\cite{Szarek1997} and~\cite[\S 8.3.3]{KSV}.
Only \cref{lem:sqrt} will be used in later sections.
Readers who are familiar with the intrinsic metric induced by operator norm
may quickly proceed to the next subsection.

Consider the operator norm on the Lie algebra~$\uu(\dd)$ of all $\dd$-by-$\dd$ antihermitian matrices.
By demanding left- and right-invariance 
we obtain a metric on a Lie group~$\UU(\dd)$,
as defined in \cref{def:pu-lie-norm}.
Recalling this definition,
the length of a smooth path $\gamma : [0,1] \to \UU(\dd)$ is given by 
\begin{equation}
    \int_0^1  \opnorm{\gamma'(t) \gamma(t)^{-1}} \rd t = \int_0^1  \opnorm{\gamma'(t)} \rd t
\end{equation}
where $\gamma'(t)\gamma(t)^{-1} \in \uu(\dd)$.
The distance between two points of $\UU(\dd)$ 
is the infimum of the lengths of all smooth paths connecting the two:
\begin{equation}
    \dist(U,V) = \inf_{\text{paths }\gamma} \int_0^1  \opnorm{\gamma'(t)} \rd t .
\end{equation}
This definition makes it obvious that $\dist$ is a metric, obeying the triangle inequality.
Though $\dist$ is equivalent to the ``extrinsic'' metric $\opnorm\cdot$,
we will use the intrinsic metric in this section because it leads us
to think in terms of the Lie algebra.

For any $U \in \UU(\dd)$ and any real number $r > 0$,
we define an open metric ball of radius~$r$ centered at~$U$ by
\begin{equation}
    \BB(r)U = \{ V \in \UU(\dd) \mid \dist(U, V) < r \} 
    = \{W U \in \UU(\dd) \mid \dist(W,\id) < r \}
\end{equation}
As the notation suggests, due to the right invariance of the metric,
the ball~$\BB(r)U$ is a shift of $\BB(r)\id$.
We will write~$\BB(r)$ for~$\BB(r)\id$.

Every unitary quantum channel is defined by a unitary,
but specifies the unitary only up to a global phase factor.
This motivates us to consider projective unitary groups~$\UU(\dd) / \uone$
where $\uone$ is the center of~$\UU(\dd)$ consisting of phase factors.
The dimension of a matrix in the center~$\uone$ is implicit.
The metric $\dist$ induces a metric on the projective unitary group 
for which we use the same notation.
In analogy with $\BB$, we write
\begin{equation}
    \PPBB(r) = \{ V \in \UU(\dd)~|~\dist(\uone\id,V) < r \}.
\end{equation}
Strictly speaking, this is not a metric ball of the projective unitary group
since $\PPBB(r)$ is a subset of~$\UU(\dd)$,
but this will hardly cause any confusion below.

We consider the fractional power of a unitary in a small neighborhood of~$\id$
for any real~$r > 0$.
The small neighborhood is actually~$\BB(\pi)$
and we define
\begin{align}
    (e^X)^{r} = e^{r X} \in \UU(\dd)
\end{align}
for any~$X \in \uu(\dd)$ with~$\opnorm X < \pi$,
which we remark is a strict inequality.
This is a proper definition 
because the exponential map is injective in the corresponding neighborhood 
of~$0 \in \uu(\dd)$; see \cref{lem:Szarek} below.
We will use the following lemma in our accuracy boosting algorithm
in the next subsection.

\begin{lemma}\label{lem:sqrt}
    For any $U, V \in \PPBB(\tfrac{1}{3\pi})$ and $p \geq 1$,
    \begin{equation} \label{eq:root-lipschitz-proj}
        \dist(U^{1/p}, V^{1/p}\uone) \leq \frac{\pi^2}{2p} \dist(U, V\uone).
    \end{equation}
\end{lemma}

The proof of this lemma will use the following.

\begin{lemma}[Following Eq.\ 5, Lemma 3, and Lemma 4 of \cite{Szarek1997}]\label{lem:Szarek}
For any~$U \in \UU(\dd)$, there is~$X \in \uu(\dd)$ with~$\opnorm X \le \pi$
such that~$U = e^X$. If~$e^X = e^Y$ with~$\opnorm X , \opnorm Y < \pi$, then~$X=Y$.
Further, the following hold:
\begin{enumerate}[label=\emph{(\alph*)},ref=\ref{lem:Szarek}(\alph*)]
\item For any~$X \in \uu(\dd)$ such that~$\opnorm X \le \pi$,
$\dist(\id, e^X) = \opnorm{X}$. \label{eq:dist-by-eigenvalues}

\item For any~$U,V \in \UU(\dd)$,
    $\opnorm{U-V} \le \dist(U,V) \le \frac \pi 2 \opnorm{U - V}$. \label{eq:metricEq}

\item
For any $X,Y \in \uu(\dd)$,
    $\opnorm{e^X - e^Y} \le \opnorm{X-Y}$. \label{eq:exp-bound}

\item 
For any $X, Y \in \uu(\dd)$ such that $\opnorm X , \opnorm Y \le \frac{1}{\pi}$,
    $\opnorm{X-Y} \le \pi \opnorm{e^X - e^Y}$. \label{eq:inv-exp-bound}
\end{enumerate}
\end{lemma}
\noindent
We did not optimize the constant~$\pi$ in~\cref{eq:inv-exp-bound}.

\begin{proof}[Proof of~\cref{lem:sqrt}]
    We first prove the non-projective version of the statement: consider $U,V \in \BB(1/\pi)$ and $p \geq 1$.
    Write $U = e^X$ and $V = e^Y$ with $\opnorm X, \opnorm Y < 1/\pi$.
    Then
    \begin{multline} \label{eq:root-lipschitz}
        \dist(U^{1/p}, V^{1/p})
        \le \opnorm{U^{1/p} - V^{1/p}}
        = \opnorm{e^{X/p} - e^{Y/p}}
        \le \opnorm{\tfrac 1 p X - \tfrac 1 p Y} \\
        = \tfrac 1 p \opnorm{X - Y}
        \le \tfrac \pi p \opnorm{e^X - e^Y}
        \le \frac{\pi^2}{2p} \dist(U, V).
    \end{multline}
    The inequalities follow from \cref{eq:metricEq}, \cref{eq:exp-bound}, \cref{eq:inv-exp-bound}, and \cref{eq:metricEq}, respectively.

    Now, to prove the lemma, without loss of generality let $U, V$ be their representatives in $\BB(\tfrac{1}{3\pi})$, and let $e^{\ii\theta}$ be the global phase minimizing $\dist(U, V\uone)$.
    Then
    \begin{align}
        \abs{\theta} = \dist(\id, e^{\ii\theta}\id)
        &\leq \dist(\id, UV^\dag) + \dist(UV^\dag, e^{\ii\theta}\id) \\
        &\leq \dist(U, V) + \dist(U, V\uone)
        \leq 2\dist(\id, UV^\dag)
        \leq \tfrac{4}{3\pi} \nonumber,
    \end{align}
    So, $Ue^{-\ii\theta/2}, Ve^{\ii\theta/2} \in \BB(1/\pi)$, so we can use \cref{eq:root-lipschitz}.
    \begin{align}
        \dist(U^{1/p}, V^{1/p}\uone)
        &\leq \dist(U^{1/p}, V^{1/p}e^{\ii\theta/p})
        = \dist((Ue^{-\ii\theta/2})^{1/p}, (Ve^{\ii\theta/2})^{1/p}) \\
        &\leq \frac{\pi^2}{2p}\dist(Ue^{-\ii\theta/2}, Ve^{\ii\theta/2})
        = \frac{\pi^2}{2p}\dist(U, V\uone) \nonumber \qedhere
    \end{align}
\end{proof}

\begin{proof}[Proof of~\cref{lem:Szarek}]
The eigenvalues of a unitary are on the complex unit circle,
each of which can be specified uniquely by an angle in~$[-\pi,\pi)$.
So, the exponential map $\uu(\dd) \ni X \mapsto e^X \in \UU(\dd)$ is
injective on the domain where $\opnorm X < \pi$.
This proves the first two statements.

\cref{eq:dist-by-eigenvalues}:
Let $\theta_j \in [-\pi,\pi)$ be the eigenvalues of~$-\ii X$.
Without loss of generality, suppose $\ell = \abs{\theta_1} \ge \abs{\theta_j}$ for all~$j$.
Then the path $\gamma$ taking $t \mapsto \exp( \ii t \diag(\theta_1,\ldots,\theta_\dd))$ 
goes from~$\id$ to~$e^X$ and has length $\ell$.
This shows that $\dist(\id, e^X) \leq \ell$.

To show that $\dist(\id, e^X) \geq \ell$, 
consider a smooth path $\gamma :[0,1] \to \UU(\dd)$ joining~$\id$ and~$e^X$.
The first column of a unitary matrix is a $\dd$-complex-dimensional unit vector, 
which corresponds to a point in the standard unit sphere $S^{2\dd-1} \subset \CC^{\dd}$.
So, the first column of the matrix~$\gamma(t)$, denoted $[\gamma(t)]_1$, defines a smooth curve~$\xi: [0,1] \to S^{2d-1}$.
Since $\opnorm{A} = \sup_{v : \twonorm{v} = 1} \twonorm{A v} \ge \twonorm{A_1}$ for any matrix~$A$, we conclude that the length of~$\gamma$ is at least the length of the path~$\xi$ under the standard Euclidean metric:
\begin{equation}
    \int_0^1 \opnorm{\gamma'(t)} \rd t
    \geq \int_0^1 \twonorm{[\gamma'(t)]_1} \rd t
    = \int_0^1 \twonorm{\xi'(t)} \rd t.
\end{equation}
The path~$\xi$ connects the two points $(1,0,\ldots,0), (e^{\ii\theta_1},0,\ldots,0) \in S^{2d-1}$.
It is well known that the length of~$\xi$ is at least the (smaller) angle between $1$ and $e^{\ii\theta_1}$ in the complex plane,
which is~$\ell$.%
\footnote{
    Since the standard sphere is a closed Riemannian manifold,
    one can appeal to the Hopf--Rinow theorem 
    to obtain a geodesic realizing the distance between any pair of points,
    and characterize geodesics~$g$ by the geodesic equation~$\nabla_{\dot g} \dot g = 0$,
    to conclude that a path of minimum length between any pair of points must be on a great circle.
}

\cref{eq:metricEq}:
By unitary invariance of the two metrics, it suffices to prove the statement for $V = \id$.
For~$\theta \in (-\pi,\pi]$, we observe that $\abs{1-e^{\ii\theta}}^2 = 2-2\cos \theta = 4 \sin^2 \frac \theta 2$,
implying that
$
    \abs{1 - e^{\ii\theta}} = \abs{2\sin(\tfrac \theta 2)} \le \abs{\theta} \le \abs{\pi\sin(\tfrac \theta 2)} = \frac \pi 2 \abs{1 - e^{\ii\theta}}.
$
It follows by~\cref{eq:dist-by-eigenvalues} that $\opnorm{\id - U} \le \dist(\id,U) \le \frac \pi 2 \opnorm{\id - U}$.

\cref{eq:exp-bound}:
\begin{align}
\opnorm{ e^{-Y} e^X - \id } = \opnorm[\Big]{ \int_0^1 \rd t\, \diff_t e^{-t Y} e^{t X} }
\le \int_0^1 \rd t \, \opnorm{e^{-t Y}(-Y +X)e^{t X}} = \opnorm{X-Y}.
\end{align}

\cref{eq:inv-exp-bound}:
Note that for any $X,Y \in \uu(\dd)$, we have
\begin{equation}
    \opnorm{e^X e^Y - e^{X+Y}} \le \frac {1} 2 \opnorm{ [X,Y] } . \label{eq:bch}
\end{equation}
This is a well-known inequality;
see {\em e.g.} Eq.~(143) in arXiv version of~\cite{CSTWZ19}.
To prove it, notice that
$
    \opnorm{e^{X} Y e^{-X} -Y }     
    = \opnorm{\int_0^1 \rd s\, \diff_s[e^{s X} Y e^{-s X}] }    
    \le \int_0^1 \rd s\, \opnorm{e^{s X} (X Y - Y X) e^{-s X} } = \opnorm{[X,Y]}
$,
so therefore,
\begin{align}
    \opnorm{e^X e^Y - e^{X+Y}} 
    &=
    \opnorm{e^{-(X+Y)} e^X e^Y - \id } = \opnorm[\Big]{\int_0^1 \rd t\, \diff_t e^{-t(X+Y)}e^{t X}e^{t Y} } \\
    &\le
    \int_0^1 \rd t\,  \opnorm*{ e^{-t(X+Y)}( -(X+Y) + X + e^{t X} Y e^{-t X} ) e^{t X}e^{t Y} } \nonumber\\
    &=
    \int_0^1 \rd t\,  \opnorm*{ -Y + e^{t X} Y e^{-t X} } \le \int_0^1 \rd t\, \opnorm{[tX, Y]} = \opnorm{[X,Y]}\int_0^1 t\,\rd t. \nonumber
\end{align}
Then,
\begin{align}
    \opnorm{X - Y}
    &= \dist(\id, e^{X-Y}) &\text{by \cref{eq:dist-by-eigenvalues}}\nonumber\\
    &\leq \frac{\pi}{2}\opnorm{\id - e^{X-Y}} \leq \frac{\pi}{2}\left( \opnorm{\id - e^{-Y}e^X} + \opnorm{e^{-Y}e^X - e^{X-Y}} \right) &\text{by \cref{eq:metricEq}}\nonumber\\
    &\leq \frac{\pi}{2}\left(\opnorm{\id - e^{-Y}e^X} + \frac12 \opnorm{X Y - Y X} \right) &\text{by \cref{eq:bch}}\nonumber \\
    &= \frac{\pi}{2}\left(\opnorm{e^Y - e^X} + \frac12\opnorm{X Y - Y Y + Y Y - Y X}\right) \\
    &\leq \frac{\pi}{2}\left(\opnorm{e^Y - e^X} + \opnorm Y \opnorm{X - Y} \right)  \nonumber\\
    &\leq \frac{\pi}{2}\left(\opnorm{e^Y - e^X} + \frac 1 \pi \opnorm{X - Y}\right) &\text{by assumption.}\nonumber
\end{align}
Rearranging, we complete the proof.
\end{proof}

\subsection{Bootstrap algorithm}

Using \cref{lem:sqrt}, we can show that we can bootstrap a base tomography algorithm that achieves constant error to an algorithm that gets $\eps^{-1}\log(1/\eta)$ query complexity.
Following \cref{sec:geometry}, for this section we use the distance metric between two unitaries $U, V \in \UU(\dd)$ of $\dist(U, V\UU(1))$ (\cref{def:pu-lie-norm}).
By \cref{prop:norm-equivalence}, this is equivalent to diamond-norm distance up to universal constants.

\begin{theorem} \label{thm:error-boosting}
    Suppose we have an oracle capable of applying an unknown unitary channel $Z \in \UU(\dd)$.
    Further suppose we have an algorithm $\alg$ that, given such an oracle, can output a unitary $\bU$ such that $\dist(Z, \bU\uone) \leq \eps_0 \leq \frac{1}{600}$ with probability $> 0.51$.
    Then, given error parameters $\eps, \eta \in (0,1)$, \cref{algo:theboot} outputs a unitary $\bU$ such that
    \begin{enumerate}[label=(\alph*),ref=\thetheorem (\alph*)]
        \item $\dist(Z, \bU\uone) \leq \eps$ with probability $\geq 1-\eta$;\label{thm:error-boosting-highprob}
        \item and $\E[\dist(Z, \bU\uone)^2] \leq (1+32\eta)\eps^2$.\label{thm:error-boosting-expectation}
    \end{enumerate}
    \cref{algo:theboot} has the further property that, if $\alg$ uses $Q$ queries to $Z$, then \cref{algo:theboot} requires only $O(\frac{Q}{\eps}\log\frac{1}{\eta})$ queries to $Z$.
\end{theorem}

By plugging in the base tomography algorithm from \cref{thm:base} into this bootstrap algorithm, we obtain our main result~\cref{thm:upperbound}.
We restate the theorem now.

\upperbound*

\begin{proof}[Proof of \cref{thm:upperbound}]
    It is clear that our bootstrap algorithm needs as many quantum registers as the base tomography algorithm does.
    The base tomography algorithm combines $O(\dd)$ pure state tomography outcomes, each of which uses projective measurements in computational basis on a pure state of form $A (ZV^\dag)^p B \ket 0$ where $A,B,V$ are unitaries known at the moment of the state tomography. See~\cref{thm:base}.
    
    The query complexity and the accuracy guarantees follow from~\cref{thm:error-boosting}, with \cref{thm:base} as the base algorithm.
    
    The number of one- and two-qubit gates used in the overall algorithm beyond the oracle calls,
    is determined by the complexity of implementing interspersing unitaries $A,B,V$ in the preparation of states~$A (ZV^\dag)^p B \ket 0$. 
    Those unitaries are on a $\dd$-dimensional qudit,
    so it can be compiled to accuracy~$\alpha$ in operator norm 
    by $\poly(\dd, \log(1/\alpha))$ elementary gates~\cite{KSV}.
    We only need $\alpha = 1/\poly(d/\eps)$.
    The algorithm prepares $\poly(\dd/\eps)$ states of form $A (ZV^\dag)^p B \ket 0$ in total.

    The classical time complexity is bounded by those of 
    matrix multiplications, matrix diagonalizations,
    and finding diagonal phase factors in \cref{prop:weaker-base}.
    All of these take $\poly(\dd/\eps)$ time.

    The bound on $\diamnorm{\calM - \calU(Z)}$ follows from \Cref{thm:error-boosting}, which shows that \cref{algo:theboot} gives $\E_\bX \opnorm{\bX}^2 \le O(\eps^2)$, and \cref{prop:conv-to-SAR}.
    Rescaling $\eps$ if necessary, we complete the proof.
\end{proof}

\begin{algorithm}
    \SetKwInOut{Input}{Input}\SetKwInOut{Output}{Output}
    \caption{Final algorithm for unitary tomography\label{algo:theboot}}
    \Input{Oracle able to black-box apply an unknown unitary~$Z$;\\
    Error parameters $\eps, \eta \in (0,1)$;\\
    Process tomography algorithm $\alg: (Z, \eta) \mapsto \bU \in \UU(\dd)$ such that \\
    \qquad $\dist(\uone \bU, Z) < \frac{1}{200}$ with probability $\geq 1-\eta$, \\
    \qquad using $O(Q\log\frac{1}{\eta})$ queries to $Z$}
    \Output{A unitary $\bU$}
    Let $T \leftarrow \lceil \log_2(1/\eps) \rceil$\;
    Let $V_{0} \leftarrow \id$\;
    \For{$j$ from $0$ to $T$}{
        Let $p_j \leftarrow 2^j$ \tcp*{This exponential could have any base $> 1$.}
        Let $\eta_j \leftarrow \eta 8^{j-T-1}$ \tcp*{Failure probability becomes \textbf{worse} each iteration.}
        Let $\bU_j$ be the output of $\alg((Z \bV_j^\dag)^{p_j},\, \eta_j)$\label{line:call}\;      
        $\bV_{j+1} \leftarrow \bU_j^{1/p_j} \bV_j$\label{line:root}\;
    }
    Return $\bU \leftarrow \bV_{T+1}$\;
\end{algorithm}

The bootstrap algorithm (\cref{algo:theboot}) is essentially four lines long: we begin with $V_0 = \id$ as an initial estimate for $Z$, and with every iteration, we refine this estimate.
In iteration $j$, we take the current estimate $\bV_j$, and consider the ``residual'' $Z\bV_j^\dagger$: the residual is close to the identity when $\bV_j$ is close to $Z$.
We run $\alg$ to get an estimate $\bU_j$ of $(Z\bV_j^\dagger)^{p_j}$.
Though this estimate $\bU_j$ is only good to constant error, by \cref{lem:sqrt}, its $p_j$th root will be $(1/p_j)$-close to $Z\bV_j^\dagger$.
So, we can use this root to form the next estimate $\bV_{j+1} = \bU_j^{1/p_j} \bV_j$, which is $(1/p_j)$-close to $Z$.
After $\log_2(1/\eps)$ iterations, this estimate is good enough to output.
We have just given nearly a full proof of the first part of \cref{thm:error-boosting}.
We swept two details under the rug: first, we need that $(Z\bV_j^\dagger)^{p_j}$ is close enough to the identity to satisfy the conditions for \cref{lem:sqrt}, which follows from induction; and we need to account for failure.
We deal with these details below.

\begin{proof}[Proof of \cref{thm:error-boosting}]
    Using the trick in \cref{prop:confidence}, we assume that our oracle $\alg$ can boost its success probability to $\geq 1-\gamma$ with only $O(\log\frac{1}{\gamma})$ overhead to output a unitary that's correct up to $c = \frac{1}{200}$ error.

    We analyze the output of~\cref{algo:theboot} for arbitrary~$T > 0$,
    first assuming that \cref{line:call} is successful for all $j$.
    Let
    \begin{align}
        \bdelta_j = \dist(Z, \bV_j\uone) = \dist(Z\bV_j^\dagger, \uone \id)
    \end{align}
    be the error after iteration $j-1$ of the while-loop.
    We will prove $\bdelta_k < 2^{-k-4}$.
    For iteration $0$, we have $p_0 = 1$ so by our assumption on $\alg$,
    \begin{align}
        \bdelta_1 = \dist(Z, \bV_1\uone) &< c < \frac{1}{32}.
    \end{align}
    We proceed by induction in~$k$.
    First, notice that the input and output of the call to $\alg$ on the $k$th iteration are in $\PPBB(\frac{1}{3\pi})$:
    \begin{align}
        \dist((Z\bV_k^\dagger)^{p_k}, \id\uone)
        &\leq \sum_{i=0}^{p_k-1} \dist((Z\bV_k^\dagger)^{i + 1}, (Z\bV_k^\dagger)^{i}\uone)
        = p_k \bdelta_k
        < \tfrac{1}{16} \tag*{by inductive hypothesis}\\
        \dist(\bU_k, \id\uone) &\leq \dist(\bU_k, (Z\bV_k^\dagger)^{p_k}\uone) + \dist((Z\bV_k^\dagger)^{p_k}, \id\uone)
        < c + \tfrac{1}{16} < \tfrac{1}{3\pi}
    \end{align}
    So, we can apply \cref{lem:sqrt} to conclude that
    \begin{align}
        \bdelta_{k+1} &= \dist(Z, \bV_{k+1}\uone)
        = \dist(Z\bV_k^\dagger, \bU_k^{1/p_k}\uone) \\
        &\leq \frac{\pi^2}{2p_k}\dist((Z\bV_k^\dagger)^{p_k}, \bU_k\uone)
        < \frac{\pi^2}{2p_k}c
        \leq \frac{\pi^2}{400}2^{-k}
        < 2^{-k-5} \nonumber
    \end{align}
    Therefore, the final accuracy guarantee is
    \begin{align}
        \dist(Z, \bV_{T+1}\uone) < 2^{-T-5} < \eps.
    \end{align}
    The query complexity is sum of those at \cref{line:call} for all $j$:
    For any given $j$, this call queries the oracle for $(Z\bV_j^\dag)^{p_j}$ for $O(Q\log\frac{1}{\eta_j})$ times,
    which can be performed using $O(Qp_j\log\frac{1}{\eta_j})$ queries to~$Z$.
    Summing over all the iterations, we get the desired total query complexity,
    \begin{align}
        O\Big(\sum_{j=0}^T Qp_j\log\frac{1}{\eta_j}\Big)
        = O\Big(Q\sum_{j=0}^T 2^j\Big((T-j+1) + \log\frac{1}{\eta}\Big)\Big)
        = O\Big(Q 2^T \log\frac{1}{\eta}  \Big)
        = O\Big(\frac{Q}{\eps}\log\frac{1}{\eta}\Big).
    \end{align}
    Notice that $\eta_j$ is chosen so the query complexity does not incur the $\log\log(1/\eps)$ factor that comes from a naive union bound.
    The failure probability of \cref{line:call} for each~$j$ is $\eta_j$, so the probability that any of them fail is bounded by $\eta$, proving \cref{thm:error-boosting-highprob}:
    \begin{align}
        \sum_{j=0}^T \eta_j
        \leq \eta\sum_{j=0}^T 8^{-(T-j)-1}
        = \eta\sum_{j=0}^T 8^{-j-1}
        \leq \eta.
    \end{align}
    Now, we prove \cref{thm:error-boosting-expectation}.
    \cref{thm:error-boosting-highprob} only implies a bound of $\E[\dist(Z, \bU\UU(1))] \leq \eps^2 + \frac{\pi}{2}\eta$, since when the algorithm fails, the output unitary could be as bad as possible, $\pi/2$ away from $Z$.
    We prove now that \cref{algo:theboot} fails gracefully: namely, if the algorithm first fails in iteration $f$, then the output only has error $O(2^{-f}$.
    Let $f$ be the least index of the iteration where the base tomography algorithm fails, 
    so that $\Pr[f = 0] < \eta_0$, and 
    $\Pr[f = j \mid f > j-1] < \eta_j$ for $j = 1,2,3,\ldots$.
    If the algorithm never fails, we take $f = T+1$.
    The argument we described above still applies up to index~$f$, so that
    \begin{align}
        \bdelta_f = \dist(Z, \bV_{f+1}\UU(1)) < 2^{-f-4}.
    \end{align}
    We can further conclude that the output of the algorithm will be about as good of an estimate to $Z$ as $\bV_{f+1}$, since the size of the adjustment in each iteration decays exponentially regardless of failure.%
    \footnote{
        In the following computation, we use that $\dist(\id, \bU_k^{1/p_k}\uone) \leq \pi 2^{-k}$, which holds because the eigenvalues of $U^{1/p_k}$ are in $(-\pi/p_k, \pi/p_k)$ for all $U \in \UU(\dd)$ with eigenvalues in $(-\pi, \pi)$.
        However, there is a technical issue: \cref{line:root} cannot be performed if $\bU_k$ has $-1 = e^{\ii \pi}$ as an eigenvalue, since in this case, the $1/p_j$-th root of this arbitrary output is not defined.
        One can and should ignore such nongeneric behavior,
        but strictly speaking our algorithm is not defined for such cases.
        Formally, one should address this by changing the choice of discretization (\cref{foot:tedium}) to ensure that the algorithm never outputs an estimate with $-1$ as an eigenvalue.
        This can be done since it only excludes a measure-zero set.
    }
    \begin{align}
        \bdelta_{T+1} &= \dist(Z, \bV_{T+1}\UU(1)) \\
        &\leq \dist(Z, \bV_{f+1}\UU(1)) + \dist(\bV_{f+1}, \bV_{T+1}\UU(1)) \nonumber\\
        &= \dist(Z, \bV_{f+1}\UU(1)) + \dist(\id, \bU_T^{1/p_T}\bU_{T-1}^{1/p_{T-1}}\cdots \bU_{f+1}^{1/p_{f+1}}\UU(1)) \nonumber\\
        &\leq \dist(Z, \bV_{f+1}\UU(1)) + \sum_{k=f+1}^T \dist(\id, \bU_k^{1/p_k}\UU(1)) \nonumber\\
        &\leq 2^{-f-4} + \sum_{k=f+1}^T \frac{\pi}{2^k} \nonumber\\
        &\leq 2^{2-f}\nonumber
    \end{align}
    This suffices to give a bound on the quality of the output $V_{T+1}$ in expectation.
    \begin{align}
        \E[\dist(Z, \bV_{T+1}\UU(1))^2]
        = \sum_{k=0}^{T+1} \Pr[f = k]\E[\bdelta_{T+1}^2 \mid f = k]
        \leq (2^{-T-5})^2 + \sum_{k=0}^T \eta_k (2^{2-k})^2 \\
        \leq (2^{-T-5})^2 + \eta\sum_{k=0}^T 2^{-3(T-k) + 4-2k} % -3n + 2 +k
        \leq (2^{-T-5})^2 +\eta 2^{5-2T}
        \leq (1+32\eta)\eps^2 \tag*{\qedhere}
    \end{align}
\end{proof}

\begin{remark}
Our bootstrap crucially depends on the continuous group structure 
of the (projective) unitary group
since we take fractional roots of a result from base tomography
and the the fractional root needs to be in the set of candidate outputs.
It is not important whether the set of candidate outputs 
is a unitary group or a projective unitary group.
Conversely, our bootstrap works for any closed Lie subgroup of the (projective) unitary group
(\emph{e.g.}, an orthogonal group $O(\dd) \subset \UU(\dd)$)
and an obvious analog of~\cref{thm:error-boosting} is proved verbatim.
Note that if a subgroup is not closed,
then the intrinsic metric $\dist$ and the extrinsic one $\opnorm \cdot$
are not equivalent and our bootstrap will not work.
Also note that some care would be needed in the base tomography;
the Fourier or Hadamard transform in the proof of~\cref{prop:weaker-base}
might not belong to the Lie subgroup in consideration.
\end{remark}

\section{Lower bounds}

The goal of this section is to prove \cref{thm:lowerbound}.
The first subsection will prove a $\Omega(\dd^2)$ lower bound using existing results on unitary discrimination~\cite{Bavaresco2021}, which we lift to a $\Omega(\dd^2/\eps)$ lower bound in the following section through a variant of an argument proving the equivalence of the fractional and discrete query models~\cite{cgmsy09,bccks17}.

\subsection{\texorpdfstring{Bound for constant $\eps$}{Bound for constant error}}

As a warm up, we first prove a query complexity lower bound for constant $\eps$.
We begin by constructing a hard instance of the estimation problem, which is a discrete set of unitaries in which each pair of unitaries is at least distance $1/4$ apart. For this instance, the task of estimating each unitary to error $<1/8$ is equivalent to the task of simply identifying the unknown unitary.

\begin{proposition} \label{reflection-net}
    There exists a set $\net = \{ W_j \in \UU(\dd) : j = 1,\ldots, N\}$ with $N \ge \exp(\dd^2 / 64)$
    of Hermitian unitaries (i.e., the eigenvalues are $\pm 1$) 
    such that $\diamnorm{\calU(W_j) - \calU(W_k)} \ge 1/4$ for any $j \neq k$.
\end{proposition}

    The constants 64 and $1/4$ are not sharp.
    A construction of some packing net of unitaries goes back at least to~\cite{Szarek1983finite}.
    Here we construct a special net of unitaries that share the same real eigenspectrum.

\begin{proof}
    Considering $2$-by-$2$ Pauli matrices 
    embedded into a higher dimensional unitary group,
    we see that the claim is true for $\dd \le 8$.
    Let $\dd = 2r + 2 > 8$ if $\dd$ is even, or $\dd = 2r +1 > 8$ if $\dd$ is odd. 
    
    We first show the existence of a set of far apart unitaries in operator norm.
    Lemma~8 of~\cite{Haah2015} shows that there are $\ge \exp(r^2/8)$ 
    density matrices of rank~$r$ in dimension~$2r$
    where exactly half the eigenvalues are $1/r$ and the other half are zero,
    and any two different density matrices in this set are at least $1/4$-apart in trace distance.
    We may write this set as $\{ (\id_{2r} + V_j)/(2r) : j = 1,\ldots, N \}$ 
    with $N \ge \exp(r^2/8)$ 
    where $V_j \in \UU(2r)$ is a Hermitian unitary of trace zero.
    If $j \neq k$, we have
    \begin{align}
        \frac 1 4 \le 
        \frac 1 2 \onenorm*{ \frac{\id_{2r} + V_j}{2r} - \frac{\id_{2r} + V_k}{2r}}
        =
        \frac 1 {4r} \onenorm{ V_j - V_k}
        \le
        \frac 1 2 \opnorm{V_j - V_k}.
    \end{align}
    Therefore, the set $S = \{ V_j \in \UU(2r) : j =1,\ldots,N\}$ 
    consists of Hermitian unitaries that are far apart in operator norm. We now need to define unitaries that are far apart in the norm used in the statement of the proposition.

    Now, embed $S$ into $\UU(\dd)$ 
    as $V_j \mapsto W_j = V_j \oplus \id_b \in \UU(\dd)$, where $b = 1$ or $2$.
    We claim that $\{ W_j : j = 1,\ldots,N\}$ is a desired set.
    Consider $W_j$ and $W_k$ for $j \neq k$.
    Then, for all $\phi \in \uone$,
    \begin{multline}
        \opnorm{\phi W_j - W_k}
        = \opnorm{\phi \id - (V_j^\dagger V_k) \oplus \id_b}
        = \max(\opnorm{\phi\id_b - \id_b}, \opnorm{\phi \id - V_j^\dagger V_k}) \\
        \geq \max(\abs{\phi - 1}, \opnorm{\id - V_j^\dagger V_k} - \opnorm{\phi \id - \id})
        \geq \half\opnorm{\id - V_j^\dagger V_k}
        = \half\opnorm{V_j - V_k}
        \geq \tfrac14.
    \end{multline}
    Minimizing over all $\phi$ and using \cref{eq:diam-pudist}, we get that $\diamnorm{\calU(W_j) - \calU(W_k)} \geq \pudist{W_j}{W_k} \geq \tfrac14$.
    Finally, $N \ge \exp(r^2/8) \ge \exp(\dd^2/64)$ for $\dd > 8$.
\end{proof}

The following result upper bounds the average success probability of any unitary identification algorithm that makes at most $Q$ uses of the unknown unitary:

\begin{proposition}[Theorem~5 of~\cite{Bavaresco2021}]\label{ChannelLearningConfidenceBound}
Let $U,\ldots, U_N \in \UU(\dd)$ by any unitary channels.
For any unitary estimation algorithm $\alg$ 
that queries an input unitary channel $Q$ times 
and outputs an index~$\hat x$ given $U_x$ with probability $\Pr[\hat x | x]$ where $x,\hat x \in [N]$,
it holds that
\begin{align}
    \frac 1 N \sum_{x=1}^N \Pr[\hat x = x | x] \le \frac 1 N \binom{ Q + \dd^2 - 1 }{Q} .
\end{align}
\end{proposition}

Ref.~\cite{Bavaresco2021} contains a proof that relies on results from several references therein.
For the readers' convenience, we present an overview of the proof in~\cref{app:Bavaresco}.

From this, it straightforwardly follows that if we use the set $\net$ we just constructed, if $Q$ is much smaller than $d^2$, then we cannot solve the identification problem with high probability. But we can also prove something stronger:

\begin{lemma}\label{BinomBound}
    If $\frac 1 N \binom{Q + \dd^2 -1}{Q} \ge \exp(- c Q)$ 
    where $Q \ge 1$, $\dd \ge 100/c' \ge 1$ and $N \ge \exp( c' \dd^2 )$ 
    for some constants $c, c' > 0$,
    then~$Q \ge c'' \dd^2$ for some constant $c''$ that depends only on $c,c'$.
\end{lemma}

In combination with~\cref{ChannelLearningConfidenceBound}, 
this immediately implies a $\Omega(\dd^2)$ query lower bound for any unitary estimation algorithm that estimates to small constant distance with constant success probability, since one can use it to distinguish between the $\exp(\Omega(\dd^2))$ unitaries in the net from \cref{reflection-net} with constant success probability.
We do not use the full power of this lemma to draw this conclusion,
but we will need it in its full form in the next subsection.

\begin{proof}
    Suppose for a contradiction $1 \le Q \le \delta \dd^2$ for some $\delta \in (0,1)$.
    It is easy to show from Stirling's approximation that $\binom{x}{s x} \le x e^{x H(s)}$ for all $x \ge 1$ where $H(s) = -s \ln (s) - (1-s)\ln(1-s)$ is the binary entropy function.
    Then, $e^{c' \dd^2 - c Q} \le \binom{Q+\dd^2-1}{Q} \leq 
    e^{H(\frac{\delta}{1+\delta}) (1+\delta)\dd^2 } (1+\delta)\dd^2$.
    Taking logs and dividing by~$\dd^2$, we get 
    \begin{equation}
    c' - \frac{2(\ln \dd)+1}{\dd^2} \le c\delta + (1+\delta) H(\frac{\delta}{1+\delta}) .
    \end{equation}
    For $\dd \ge 100/c'$, the left-hand side is at least $c'/2$.
    This completes the proof.
\end{proof}

\subsection{\texorpdfstring{Bound for general $\eps$}{Bound for general error}}

We're now ready to prove the stronger $\eps$-dependent lower bound.
Our high-level strategy is as follows.
First, we construct a special set of unitaries that are pairwise $\eps$ apart using our constructed set $\net$. This constructed set brings all the unitaries in $\net$ closer to the identity (and hence each other) by raising each unitary to a small power $\alpha$, which will be of the order of $\eps$. 
Then assuming there is an $\alg$ that uses $Q$ queries to solve the unitary estimation problem on this set, there exists another quantum algorithm for identifying unitaries from the set $\net$ with $O(\eps Q)$ query complexity, but there is a catch: The query complexity stated is in the \emph{fractional query model}, a model introduced by Cleve, Gottesman, Mosca, Somma, and Yonge-Mallo~\cite{cgmsy09}. (In this model, it is cheaper to query a small power of a unitary compared to querying the unitary itself.) 
But we want to show our lower bound in the standard model assumed in \Cref{thm:lowerbound}. So we use a modified proof of the equivalence of the fractional and discrete query models from Berry, Childs, Cleve, Kothari, and Somma~\cite{bccks17} to give an algorithm using $O(\eps Q)$ queries in the usual sense, except with failure probability $\geq \exp(-50\eps Q)$ instead of $\frac23$. 
This extremely low success probability is not a problem, since \cref{ChannelLearningConfidenceBound} and \Cref{BinomBound} yield strong lower bounds even with small success probability. Thus any such algorithm, which used $O(\eps Q)$ queries, must use $\Omega(\dd^2)$ queries, which gives $Q = \Omega(\dd^2/\eps)$ as desired.

We start by defining some notation for the fractional power of a reflection.
Following convention from~\cite{bccks17}, this is the power in the \emph{negative angle} direction: with this definition, $(-1)^\alpha = e^{-\ii\pi\alpha}$.
\begin{definition}
For a reflection $R \in \UU(\dd)$ (i.e., $R^2 = \id$) and $\alpha \in [-1,1]$, we define
\begin{align}
    R^\alpha \coloneqq \frac12(\id + R) + e^{-\ii\pi\alpha}\frac12(\id - R)
    = e^{-\ii\pi\alpha/2}\left(\id \cos \frac{\pi\alpha}{2} + \ii R \sin \frac{\pi\alpha}{2} \right). \label{eq:PiToAlpha}
\end{align}
\end{definition}

We can now state our main conversion lemma, which takes an algorithm that queries a fractional power of $Z$ and converts it to an algorithm that only queries $Z$ and makes fewer queries, at the cost of having low success probability.

\begin{lemma} \label{lem:bad-bccks}
    Let $\circuit \in \UU(D)$ be implemented by a quantum circuit that uses $Q$ queries to $Z$ or $Z^\dagger$, with $Z = R^\alpha$ for $R \in \UU(\dd)$ a reflection and $\alpha \in (0,1]$.

    Then there is a circuit $\circuit_R \in \UU(2^Q D)$ with $Q$ ancilla qubits, 
    which uses $50 + 100\alpha Q$ queries to the controlled unitary $\controlled R$ and such that, 
    for all $\ket{\psi} \in \mathbb{C}^D$,
    \begin{gather}
        \circuit_R\ket{0}^{\otimes Q}\ket{\psi} = \nu\ket{0}^{\otimes Q}\ket{\widetilde{\mathrm{out}}} + \ket{\Phi^\perp}, \\
    \intertext{where $\ket{\widetilde{\mathrm{out}}}$ is normalized, $\nu \in \mathbb C$ is a scalar, $(\bra{0}^{\otimes Q} \otimes I) \ket{\Phi^\perp} = 0$  and }
        \abs{\nu} \geq 0.99\exp(-\alpha\pi \tfrac{Q}{2}) \qquad \text{ and }\qquad \twonorm{\ket{\widetilde{\mathrm{out}}} - \circuit\ket{\psi}} < \exp(-99\alpha \pi Q - 20).
    \end{gather}
\end{lemma}

A similar conversion was also stated in \cite[Lemma~3.8]{bccks17}, but their lemma converts the circuit $\circuit$ that uses $Q$ oracle calls to $Z$ and $Z^\dagger$ to a circuit that uses $O(\alpha Q \frac{\log(\alpha Q/\xi)}{\log\log(\alpha Q/\xi)})$ oracle calls to $\controlled R$ and incurs $\ell^2$-norm error~$\xi$ in the output state.
Here, $\alpha Q$ is the fractional query complexity of the circuit, since in that model each application of $Z$ and $Z^\dagger$ incurs $\alpha$ cost. So their lemma is similar to ours, but incomparable: It can achieve much smaller error, but the query complexity has a log factor which would weaken our lower bound if used directly.

To eliminate this log factor, we will modify their argument and approximate the complete circuit instead of each individual segment.
This gives a circuit with $O(\alpha Q)$ calls to $\controlled R$, but with the significantly worse success probability of $0.9\exp(-\pi \alpha Q)$.
This exponentially decaying success probability is generally undesirable, but will turn out not to affect the lower bound argument, which merely needs that the success probability of the algorithm is not significantly better than guessing, 
which succeeds with probability $\frac 1 N \le \exp(-\dd^2/64)$.

\begin{proof}
    We closely follow the fractional query to discrete query reduction in \cite{bccks17}.
    We begin by writing a circuit to implement $Z$ or $Z^\dagger$ from an application of $\controlled R$.
    Consider the following circuit. % shown in \cref{eq:fractional-power-circuit}.
    \begin{align} \label{eq:fractional-power-circuit}
        \begin{array}{ccc}
        \begin{array}{c}\Qcircuit @C=1em @R=1.2em {
            \lstick{\ket{0}} & \gate{P^\pm_\alpha} & \ctrl{1} & \gate{P^\pm_\alpha} & \qw \\
            \lstick{\ket{\psi}} & \qw & \gate{R} & \qw & \qw
        }\end{array}
        &
        \begin{array}{c} P^\pm_\alpha \coloneqq \begin{bmatrix} \sqrt{\gamma} & \pm \ii \sqrt{1-\gamma} \\ \sqrt{1-\gamma} & \mp\ii\sqrt{\gamma}\end{bmatrix}\end{array}
        &
        \begin{array}{c} \gamma \coloneqq \frac{\cos\frac{\alpha\pi}{2}}{\cos \frac{\alpha\pi}{2} + \sin\frac{\alpha\pi}{2}} >0\end{array}
        \end{array}
    \end{align}
    Let us calculate the evolution of the state vector along this circuit:
    \begin{align} \label{eq:gadget-computation}
        &\ket{0}\ket{\psi} \\
        &\mapsto (\sqrt{\gamma}\ket{0} + \sqrt{1-\gamma}\ket{1})\ket{\psi} \nonumber\\
        &\mapsto \sqrt{\gamma}\ket{0}\ket{\psi} + \sqrt{1-\gamma}\ket{1}R\ket{\psi} \nonumber\\
        &\mapsto \sqrt{\gamma}(\sqrt{\gamma}\ket{0} + \sqrt{1-\gamma}\ket{1})\ket{\psi} \pm \ii \sqrt{1-\gamma}(\sqrt{1-\gamma}\ket{0} - \sqrt{\gamma}\ket{1})R\ket{\psi}\nonumber\\
        &= \ket{0}(\gamma\ket{\psi} \pm \ii (1-\gamma)R\ket{\psi}) + \sqrt{\gamma(1-\gamma)}\ket{1}(\ket{\psi} \mp \ii R\ket{\psi})\nonumber\\
        &= \ket{0}\frac{\cos(\half\alpha\pi)\ket{\psi} \pm \ii\sin(\half\alpha\pi)R\ket{\psi}}{\cos(\half{\alpha}\pi) + \sin(\half{\alpha}\pi)} + \sqrt{\gamma(1-\gamma)}\ket{1}(\ket{\psi} \mp \ii R\ket{\psi}) \nonumber\\
        &= \frac{e^{\ii\pi\alpha/2}}{{\cos(\half{\alpha}\pi) + \sin(\half{\alpha}\pi)}}\ket{0}R^{\pm\alpha}\ket{\psi} 
        + \sqrt{\gamma(1-\gamma)}\ket{1}(\ket{\psi} \mp \ii R\ket{\psi}) \nonumber
    \end{align}
    
    Using the circuit \cref{eq:fractional-power-circuit} with either $P^+ _\alpha$ or $P^-_\alpha$, 
    we can apply $R^\alpha = Z$ or $R^{-\alpha} = Z^\dagger$, respectively, with one use of $\controlled R$
    and postselection of the ancilla on~$\ket 0$.
    This reproves \cite[Lemma~3.3]{bccks17} in the slightly more general setting 
    where we may implement $Z^\dagger$ and the reflection~$R$ is arbitrary, not necessarily diagonal.

\begin{figure}[ht]
\[ 
\Qcircuit @R=.6em @C=0.65em {
\push{\ket{0}} \gategroup{1}{1}{4}{2}{.5em}{.} & \gate{P_\alpha^{\pm}} & \qw & \ctrl{4} & \qw & \qw &{\cdots}& & \qw & \qw& \gate{P_\alpha^{\pm}}  & \qw  \\
\push{\vdots} &        &    &   & &   &\ddots& &  &    & \vdots &   \\
   &        &        &   & &   & & &  &    &  &   \\
\push{\ket{0}} & \gate{P_\alpha^{\pm}} & \qw & \qw & \qw & \qw &{\cdots}& & \qw & \ctrl{1} & \gate{P_\alpha^{\pm}} & \qw \\
\push{\hphantom{\ket{\psi}}}& \qw & \multigate{3}{U_0} & \multigate{3}{R} & \multigate{3}{U_1} & \qw &{\cdots}& & \multigate{3}{U_{m-1}} & \multigate{3}{R}& \multigate{3}{U_{m}} &  \qw  \\
\push{\ket{\psi}}& {\vdots} & &  &  &  &\ddots& &  &  &  &    \\
& & &  &  &  && &  &  &  &   \\
\push{\hphantom{\ket{\psi}}} & \qw & \ghost{U_1} & \ghost{Q} & \ghost{U_1} & \qw &{\cdots}& & \ghost{U_{Q-1}} & \ghost{Q}& \ghost{U_{Q}} &    \qw  \\
}
\]
\caption{
The augmented version $\widehat{\circuit}$ of the circuit $\circuit$, which alternates between unitaries $U_i$ and an application of $Z$ or $Z^\dagger$, to perform $Q$ queries in total.
Each application of $Z$ or $Z^\dagger$ is replaced with the gadget in~\cref{eq:fractional-power-circuit}, making for $Q$ additional ancilla qubits.
The dotted box indicates the state prepared in the ancilla, denoted $\ket{\anc}$ in \cref{eq:anc-def}.
Figure taken from \cite[Figure~2]{bccks17} with slight modifications.
}
\label{fig:frac-reduction}
\end{figure}
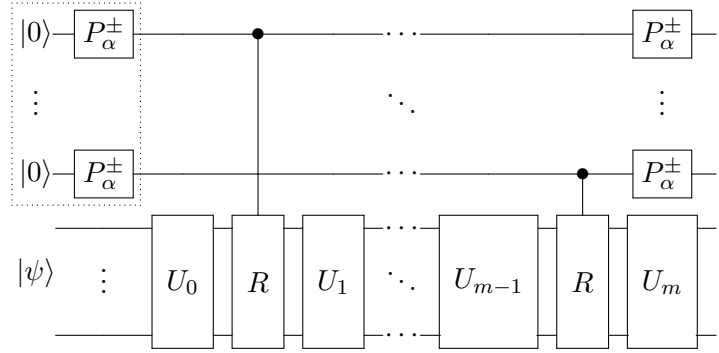

    Let $\widehat{\circuit}$ denote the circuit $\circuit$ augmented so that every instance of $Z$ and $Z^\dagger$ is replaced with the circuit 
    in~\cref{eq:fractional-power-circuit} (where we suppose $\alpha > 0$) as shown in \cref{fig:frac-reduction}.
    By composing \cref{eq:gadget-computation}, we get that
    \begin{align} \label{eq:failure}
        \widehat{\circuit}\ket{0}^{\otimes Q}\ket{\psi} &= 
        \underbrace{\Big(\frac{e^{\ii\pi\alpha/2}}{{\cos(\half{\alpha}\pi) + \sin(\half{\alpha}\pi)}}\Big)^Q}_\nu\ket{0}^{\otimes Q}\circuit\ket{\psi} + \ket{\Phi^\perp},
    \end{align}
    where $(\bra{0}^{\otimes Q} \otimes I) \ket{\Phi^\perp} = 0$, and the amplitude $\nu$ on the desired state satisfies
    \begin{align}
        \abs{\nu}
        &= (\cos(\half\alpha\pi) + \sin(\half\alpha\pi))^{-Q}
        = (1 + \sin(\alpha\pi))^{-\frac{Q}{2}}
        \geq \exp(-\sin(\alpha\pi)\tfrac{Q}{2})
        \geq \exp(-\alpha\pi\tfrac{Q}{2}).
    \end{align}
    Now it seems like we have made no progress, since the augmented algorithm $\widehat{\circuit}$ naively uses $Q$ queries to $\controlled R$, which is the same number of queries used by the original algorithm. We would like to reduce this query complexity to $100 \alpha Q$. 
    If $\alpha > \half$, the $\circuit_R = \widehat{\circuit}$ is what we want. 
    So, assume $0 < \alpha \leq \half$.

    The observation that lets us drastically reduce the query complexity of this algorithm is that the initial state of the ancilla has very low weight on $Q$-bit strings with a large number of $1$'s.
    In other words, it is mostly supported on low Hamming weight strings. Observe that the state is
    \begin{align} \label{eq:anc-def}
        \ket{\anc}
        = \bigotimes_{t=1}^Q P^\pm_\alpha \ket{0}
        = \bigotimes_{t=1}^Q (\sqrt{\gamma}\ket{0} + \sqrt{1-\gamma}\ket{1}),
        \quad \text{where }\gamma = \frac{\cos(\half\alpha\pi)}{\cos(\half\alpha\pi) + \sin(\half\alpha\pi)}.
    \end{align}
    Here, we think of $\alpha$ as being small, making $1-\gamma = \Theta(\alpha)$. Indeed,
    \begin{align} \label{eq:gamma-bounds}
        \tfrac14\alpha\pi \leq 1-\gamma \leq \tfrac12\alpha\pi \text{ for } \alpha \in [0,\half].
    \end{align}
    If $\ket{\anc}$ were replaced by an approximate state $\ket{\widetilde{\anc}}$ that is a superposition over \emph{only} the bit strings of Hamming weight $\leq K$, if $K$ is large enough (roughly $O(\alpha Q)$), then the approximate state will be very close to the original state, and hence will not affect the output too much. 
    
    Then notice what the resulting circuit looks like: It has $Q$ c$R$ gates, but the control register never holds a string of Hamming weight larger than $K$. In other words, in no branch of the superposition do we ever apply $R$ more than $K$ times, although the circuit formally has $Q$ copies of c$R$. So morally, this circuit should only require $K$ queries to $R$, not $Q$ queries.

    The argument for this is given more explicitly in \cite[Page 28]{bccks17}, but briefly, a bit string in the superposition can be thought of as encoding a partition of the circuit, dictating how long the circuit should be run before a query from $\alg$ is responded with $R$ instead of $\id$; in this way, one can run the circuit with $K$ oracle queries, by conditioning the intermediate unitaries on the ancilla bit string.
    
    Let $\ket{\anc_{\leq K}}$ denote the subnormalized state of $\ket{\anc}$
    obtained by discarding all computational basis components
    of Hamming weight $> K$.
    We consider a cutoff $K = k(1-\gamma) Q$, where $(1-\gamma) Q$ is the ``mean'' bitstring weight, and $k$ is a parameter to be chosen.
    Then
    \begin{align}
        \ket{\anc} - \ket{\anc_{\leq K}}
        &= \sum_{\substack{b \in \{0,1\}^Q \\ \abs{b} > K}} \sqrt{\gamma}^{Q-\abs{b}}\sqrt{1-\gamma}^{\abs{b}}\ket{b}, \quad \textrm{and} \\
        \twonorm{\ket{\anc} - \ket{\anc_{\leq K}}}
        &= \Big(\sum_{\substack{b \in \{0,1\}^Q \\ \abs{b} > K}} \gamma^{Q-\abs{b}}(1-\gamma)^{\abs{b}}\Big)^{\frac12}
        < \exp\Big(-\frac{k^2(1-\gamma)Q}{2(2+k)}\Big),
    \end{align}
    where we used a Chernoff bound since the error can be written as the probability that a binomial random variable takes a large value.
    With this, we can compute the $\ell^2$ error incurred by replacing $\ket{\anc}$ with a normalized state
    $\ket{\anc_{\leq K}}/\twonorm{\ket{\anc_{\leq K}}}$. 
    \begin{align}
        \twonorm[\Big]{\ket{\anc} - \frac{\ket{\anc_{\leq K}}}{\twonorm{\ket{\anc_{\leq K}}}}}
        &\leq \twonorm[\Big]{\ket{\anc} - \ket{\anc_{\leq K}}} + \twonorm[\Big]{\ket{\anc_{\leq K}} - \frac{\ket{\anc_{\leq K}}}{\twonorm{\ket{\anc_{\leq K}}}}} \label{eq:ancilla-error}\\
        &= \twonorm{\ket{\anc} - \ket{\anc_{\leq K}}} + 1 - \twonorm{\ket{\anc_{\leq K}}} \nonumber\\
        &\leq 2\twonorm{\ket{\anc} - \ket{\anc_{\leq K}}}
        < 2\exp\Big(-\frac{k^2(1-\gamma)Q}{2(2+k)}\Big) \nonumber
    \end{align}
    We take $K = \ceil{40 + 40(1-\gamma)Q}$, so $k \geq 40 + \frac{40}{(1-\gamma)Q}$, and
    \begin{align}
        2\exp\parens[\Big]{-\frac{k^2(1-\gamma)Q}{2(2+k)}}
        &\leq \exp\parens[\Big]{\ln(2)-\frac{k(1-\gamma)Q}{4}} \\
        &\leq \exp\parens[\Big]{\ln(2)-\frac{20^2(1-\gamma)Q + 100}{4}}
        \leq \exp\parens[\Big]{-100 \alpha \pi Q - 24} \tag*{by \cref{eq:gamma-bounds}}
    \end{align}
    So, let $\widehat{\circuit}_{\leq K}$ denote the circuit of $\widehat{\circuit}$, modified so that upon being given $\ket{0}^{\otimes Q}$, the circuit prepares the truncated $\ket{\anc_{\leq K}}/\twonorm{\ket{\anc_{\leq K}}}$ instead of $\ket{\anc}$.
    This will be our final choice of circuit $\circuit_R$, so our goal now is to establish the claimed properties of the resulting circuit
    \begin{align}
        \widehat{\circuit}_{\leq K} \ket{0}^{\otimes Q}\ket{\psi}
        = \nu_{\leq K}\ket{0}^{\otimes Q}\ket{\widetilde{\mathrm{out}}} + \ket{\Phi^\perp},
    \end{align}
    where $\ket{\Phi^\perp}$ is some state satisfying $(\bra{0}^{\otimes Q} \otimes I) \ket{\Phi^\perp} = 0$.
    
    There is some phase ambiguity in definitions, since $\ket{\widetilde{\mathrm{out}}}$ could have any phase as long as $\nu_{\leq K}$ compensates for it so that the product of the phases is correct. So we will choose $\nu_{\leq K}$ to have the same phase as $\nu$, so that $\abs{\nu_{\leq K} - \nu} = \abs{\abs{\nu_{\leq K}} - \abs{\nu}}$.
    If the ancilla state changes by $\xi$ in $\ell^2$ distance, since the rest of the circuit is unitary, this changes the output state of $\widehat{\circuit}$ by $\xi$ error.
    Using this, we can conclude
    \begin{align} \label{eq:nu-correctness}
        \abs{\abs{\nu} - \abs{\nu_{\leq K}}}
        &= \abs{\twonorm{(\bra{0}^{\otimes Q}\otimes I)\widehat{\circuit}\ket{0}^{\otimes Q}\ket{\psi}} - \twonorm{(\bra{0}^{\otimes Q}\otimes I)\widehat{\circuit}_{\leq K}\ket{0}^{\otimes Q}\ket{\psi}}} \\
        &\leq \twonorm{(\bra{0}^{\otimes Q}\otimes I)(\widehat{\circuit} - \widehat{\circuit}_{\leq K})\ket{0}^{\otimes Q}\ket{\psi}} \nonumber\\
        &\leq \twonorm{(\widehat{\circuit} - \widehat{\circuit}_{\leq K})\ket{0}^{\otimes Q}\ket{\psi}} \leq \exp\parens{-100 \alpha \pi Q - 24}. \nonumber
    \end{align}
    A consequence of this is that $\abs{\nu_{\leq K}} \geq \exp(-\alpha\pi\tfrac{Q}{2}) - 0.01\exp(-100\alpha\pi Q) \geq 0.99\exp(-\alpha\pi\tfrac{Q}{2})$.
    Similarly, we can show
    \begin{align}
        \twonorm{\circuit\ket{\psi} - \ket{\widetilde{\mathrm{out}}}}
        &\leq \frac{1}{\nu}(\twonorm{\nu\circuit\ket{\psi} - \nu_{\leq k}\ket{\widetilde{\mathrm{out}}}} + \abs{\nu_{\leq k} - \nu}\twonorm{\ket{\widetilde{\mathrm{out}}}}) \\
        &=\frac{1}{\nu}(\twonorm{(\bra{0}^{\otimes Q}\otimes I)(\widehat{\circuit} - \widehat{\circuit}_{\leq K})\ket{0}^{\otimes Q}\ket{\psi}} + \abs{\abs{\nu_{\leq k}} - \abs{\nu}}\twonorm{\ket{\widetilde{\mathrm{out}}}}) \\
        &\leq\frac{2}{\nu}\twonorm{(\bra{0}^{\otimes Q}\otimes I)(\widehat{\circuit} - \widehat{\circuit}_{\leq K})\ket{0}^{\otimes Q}\ket{\psi}} \leq \frac{2}{\nu}\exp\parens{-100 \alpha \pi Q - 24} \tag*{by \cref{eq:nu-correctness}}\\
        &\leq 2\exp(\alpha \pi\tfrac Q2)\exp\parens{-100 \alpha \pi Q - 24}
        < \exp\parens{-99 \alpha \pi Q - 20}. \tag*{by \cref{eq:failure}}
    \end{align}
    This gives the desired properties.
    The number of queries made is $K = \ceil{40(1 + (1-\gamma)Q)} \leq 50 + 100\alpha Q$.
\end{proof}

We are now ready to prove the final lower bound, which we restate for the reader's convenience:

\lowerbound*

\begin{proof}
    Consider a quantum circuit $\alg$ with the assumed properties and that uses $Q$ queries to the oracle, where $Q$ is a function of $\dd$ and $\eps$.
    Consider the net of reflections $\net \subset \UU(\dd)$ as described in \cref{reflection-net}.
    We let $Z = R^{\alpha}$.
    If we run a process tomography algorithm $\alg$ with the properties assumed in the statement on $Z$, then $\alg$ is promised to output an estimate~$\bU$ such that $\dist(\UU(1)\bU, Z) < \eps$ with probability $\geq 2/3$ using $Q$ queries of $\controlled Z$ and $\controlled Z^\dagger$.
    Since $(\controlled R)^\alpha = \controlled(R^\alpha) = \controlled Z$, these $\controlled Z$ and $\controlled Z^\dagger$ queries can be thought of as \emph{fractional queries} to $\controlled R$.
    We first show that the assumed output guarantee, run on an element $R \in \net$ of the net, suffices to determine which element $R$ is, upon choosing $\alpha = \Theta(\eps)$; then, we show how to perform this algorithm, only using queries to $\controlled\controlled R$, the doubly controlled~$R$.

    For the first step, suppose we have access to an unknown $R \in \net$, and wish to identify which element it is.
    We can do this by applying $\alg$ to $Z = R^\alpha$ for $\frac{1}{\alpha} = \floor{\frac{1}{8\eps}}$ (where we assume that $\eps < \frac{1}{8}$ so that $\frac1\alpha \geq 1$).
    On a successful run, which occurs with probability $\geq \frac23$, the output $\bU$ satisfies $\diamnorm{\calU(\bU) - \calU(Z)} < \eps$.
    In this case, we can identify which element in the net it is by computing $\bU^{1/\alpha}$, and picking the element of the net it is closest to.
    The correctness guarantee implies that $\bU^{1/\alpha}$ and $R$ are close in diamond norm:
    \begin{multline}
        \diamnorm{\calU(\bU^{1/\alpha}) - \calU(R)}
        = \diamnorm{\calU(\bU^{1/\alpha}) - \calU(Z^{1/\alpha})} \\
        = \diamnorm[\Bigg]{\sum_{p=1}^{1/\alpha}(\calU(\bU^{p}Z^{1/\alpha - p}) - \calU(\bU^{p-1}Z^{1/\alpha - p+1}))} 
        \leq \frac1\alpha\diamnorm{\calU(\bU) - \calU(Z)} < \frac{\eps}{\alpha} \leq \frac{1}{8},
    \end{multline}
    where above we used the triangle inequality and unitary invariance of the diamond norm.
    By the property of the net $\net$, only one element is ever within $\frac{1}{8}$ of any particular $\bU^{1/\alpha}$, so the closest net element to it must be the correct element $R$.

    Now, we show that $\alg$ can still be run given appropriate access to $R$ only.
    Applying \cref{lem:bad-bccks} on $\alg$ with $Q$ queries to $\controlled Z, \controlled Z^\dagger$, we get a circuit $\alg_R$ that uses $50+100\alpha Q$ queries to $\controlled\controlled R$.
    Further, if we run $\alg_R$, the ancilla qubits will be measured in~$\ket{0}^{\otimes Q}$ with probability $\geq 0.9\exp(-\alpha\pi Q)$, and the resulting pure state will have $\ell^2$-norm error $\exp(-99\alpha \pi Q - 20)$ from the true output of~$\alg$.
    We will perform $\alg_R$, measure all of the ancilla in the computational basis, and if it is~$\ket{0}^{\otimes Q}$ (postselection), 
    we output the output of~$\alg_R$ (which we assume to be the measurement of the rest of the qubits in the computational basis); 
    otherwise, we declare failure.
    Since $\alg$ produces a correct estimate with probability $\geq \frac23$, 
    and since closeness (in $\ell^2$ norm) of states $\ket{a}, \ket{b}$ implies closeness (in total variation distance) between the probability distributions resulting from measuring in the computational basis, $\tvdist{\mathbf{a}}{\mathbf{b}} \leq 4\twonorm{\ket{a} - \ket{b}}$~\cite[Lemma~3.6]{bv97},
    we have 
    \begin{align}
        &\Pr[\alg_R\text{ succeeds}]\\
        &= \Pr[\alg_R\text{ postselection succeeds}]\Pr[\alg_R \text{ succeeds} \mid \text{postselection succeeds}] \nonumber\\
        &\geq 0.9 e^{-\alpha\pi Q} \big(\Pr[\alg \text{ succeeds}] - \tvdist{\text{\textbf{output}\ of postselected} \alg_R}{\text{\textbf{output}\ of}\alg}\big) \nonumber\\
        &\geq 0.9 e^{-\alpha\pi Q} \big(\Pr[\alg \text{ succeeds}] - 4\twonorm[\big]{\ket{\text{output state of postselected } \alg_R} - \ket{\text{output state of }\alg}}\big) \nonumber\\
        &\geq 0.9 e^{-\alpha\pi Q} (\tfrac{2}{3} - 4 e^{-99\alpha \pi Q - 20} ) \nonumber\\
        &\geq \half e^{-\alpha\pi Q}. \nonumber
    \end{align}
    To summarize, we have an algorithm that, with $50 + 100\alpha Q$ queries to $\controlled\controlled R$, can output a classical description $\bU$ of $Z = R^\alpha$ such that $\diamnorm{\calU(\bU) - \calU(Z)} < \eps$ with probability $\geq \frac{1}{2}\exp(-\alpha\pi Q)$.

    If we are promised that $R \in \net$ with $\net$ the net from \cref{reflection-net}, then choosing $\alpha = \floor{\frac{1}{8\eps}}^{-1}$, the output of this algorithm can distinguish an element $R$ in the net with probability $\geq \frac{1}{2}\exp(-\alpha\pi Q)$.
    It can also distinguish elements of the net $\controlled\controlled\net = \{\controlled\controlled R \mid R \in \net\} \subset \UU(4\dd)$, since identifying $R$ is equivalent to identifying $\controlled\controlled R$.
    By~\cref{ChannelLearningConfidenceBound} and~\cref{BinomBound}, any algorithm that can distinguish between all the elements in $\controlled\controlled\net$ with probability $\geq \exp(\delta_2-\delta_1 Q)$, 
    where $\delta_1 \geq 0, \delta_2 \in \mathbb{R}$ are constants,
    must use $\Omega(\dd^2)$ queries.
    From this we can conclude the bound:
    \begin{equation*}
        \Omega(\dd^2) \le 50 + 100\alpha Q = 50 + 100\floor{\tfrac{1}{8\eps}}^{-1}Q
        \implies Q = \Omega(\dd^2/\eps).\qedhere
    \end{equation*}
\end{proof}

\addcontentsline{toc}{section}{References}
\bibliography{main}
\bibliographystyle{alphaurl}

\appendix

\section{\texorpdfstring{Deferred proofs for \cref{sec:norms}}{Deferred proofs for Section 1.1}} \label{sec:norm-proofs}

\normeq*

\begin{proof}
We can get \cref{eq:pudist-lie} from minimizing the terms in the inequality \cref{eq:metricEq} over a global phase.
So, it suffices to prove \cref{eq:diam-pudist}, which relates diamond-norm distance to the other distances.

Let $\lambda_1,\ldots,\lambda_\dd$ denote the eigenvalues of $U^\dagger V$, and let $\sigma$ denote their \emph{spread}, i.e.\ the length of the shortest arc (of the complex unit circle) containing the spectrum of $U^\dagger V$:
\begin{align} \label{eq:spread}
    \sigma = \max_{\substack{-2\pi \leq a \leq b \leq 2\pi \\ \{\lambda_i\} \subset \{e^{\ii\theta} \mid \theta \in [a,b]\}}} \abs{b - a}.
\end{align}
Then, examining the geometry, we can rewrite all of the distances in terms of this spread $\sigma$.%
\footnote{
    This reflects that all these norms are \emph{worst-case} norms, in that they in some sense maximize error over all possible directions.
}

For diamond-norm distance, we get a simple formula using special properties for when the two channels are isometries.
\begin{align} 
    \diamnorm*{ \calU(U) - \calU(V) }
    &= \max_{x \in \mathbb{C}^\dd}\onenorm{Ux(Ux)^\dagger - Vx(Vx)^\dagger} \tag*{by \cite[Theorem~3.55]{Watrous2018}} \\
    &= \max_{x \in \mathbb{C}^\dd} 2\sqrt{1 - \abs{x^\dagger U^\dagger Vx}^2} \tag*{by \cite[Equation~1.186]{Watrous2018}} \\
    &= 2 \sqrt{1-\min\{|z|^2 : z \in \text{convex-hull}(\lambda_1, \dots, \lambda_{\dd})\}} \label{eq:diamond-hull}\\
    &= \begin{cases}
        2\sqrt{1 - \abs{\half(1 + e^{\ii\sigma})}^2} = 2\sin(\sigma/2) & \text{ if }\sigma < \pi \\
        2 & \text{ if }\pi \leq \sigma \leq 2\pi
    \end{cases}
    \label{eq:diamond}
\end{align}
To get \cref{eq:diamond-hull}, if the spread is $\geq \pi$, then the convex hull must contain the origin, and so the minimum is $0$; otherwise, the minimum is achieved at the midpoint of the extremal eigenvalues, $\half(e^{\ii a} + e^{\ii b})$ according to the notation in \cref{eq:spread}.

We can also see
\begin{align}
    \pudist{U}{V}
    &= \min_{\phi \in \uone}\opnorm{\id - \phi U^\dagger V}
    = \min_{\phi \in \uone}\max_{i} \abs{1 - \phi \lambda_i}
    = \abs{1-e^{\ii\sigma/2}}
    %= \abs{e^{\ii\sigma/4} - e^{-\ii\sigma/4}}
    = 2\sin(\sigma/4); \label{eq:pudist-spread}\\
    \dist(U, \uone V)
    &= \min_{\phi \in \uone} \dist(\id, \phi U^\dagger V)
    = \min_{\phi \in \uone} \max_{i} \dist(1, \phi \lambda_i)
    = \sigma/2.
\end{align}
Using the following basic inequalities,
\begin{align}
    \half\sin(\sigma/2) \leq \sin(\sigma/4) &\leq \sin(\sigma/2) &\text{for }& \sigma < \pi, \\
    \half \leq \sin(\sigma/4) &\leq 1 &\text{for }& \pi \leq \sigma \leq 2\pi, \nonumber
\end{align}
we can conclude from \cref{eq:diamond,eq:pudist-spread} that \cref{eq:diam-pudist} holds:
\begin{equation*}
    \half\diamnorm{\calU(U) - \calU(V)} \leq \pudist{U}{V} \leq \diamnorm{\calU(U) - \calU(V)}. \qedhere
\end{equation*}
\end{proof}

\fiddiam*

\begin{proof}
For $\dd = 1$, there is only one unitary channel (the identity channel), so all distances are zero and the equality holds trivially.
So, suppose $\dd \geq 2$ and, as in the proof above, let $\lambda_1,\ldots,\lambda_\dd$ denote the eigenvalues of $U^\dagger V$.
Let $\sigma \in [0,2\pi)$ be the length of a shortest arc on the unit complex circle 
that covers all the eigenvalues.
By rotating if necessary, we may set $\lambda_1 = e^{\ii \sigma / 2}$ and $\lambda_\dd = e^{-\ii \sigma / 2}$.
Recall the following expressions.
\begin{align}
    \ol{F}(\calU(U),\calU(V)) &= 1 - \abs[\Big]{\tfrac{1}{\dd} \Tr(U^\dag V)}^2
    = 1 - \abs[\Big]{\tfrac{1}{\dd}\sum_i \lambda_i}^2, & (\text{\cref{eq:entfid}}) \nonumber\\
    \tfrac 1 4 \diamnorm*{ \calU(U) - \calU(V) }^2
    &= 1-\min\{|z|^2 : z \in \text{convex-hull}(\lambda_1, \dots, \lambda_{\dd})\} & (\text{\cref{eq:diamond-hull}})\\ 
    &= \sin^2 (\min(\sigma, \pi)/2). & (\text{\cref{eq:diamond}})\nonumber 
\end{align}
Since $\frac{1}{\dd}\sum_i \lambda_i$ is a point in the convex hull, 
the first two lines here give the left-hand side of the proposition.
For the right-hand side,
if $\sigma \le \pi$, then
\begin{equation}
    \abs*{\frac 1 \dd \sum_i \lambda_i } = \abs*{ \frac 2 \dd \cos \frac \sigma 2 + \frac 1 \dd \sum_{i=2}^{\dd -1} \lambda_i } 
    \le \frac 2 \dd \cos \frac \sigma 2 + \frac{\dd -2}{\dd} .
\end{equation}
If $\sigma > \pi$, then $\dd > 2$ and there must be some eigenvalue, say, $\lambda_j = e^{\ii \theta}$ 
where $\theta$ is in the interval $J := [-\frac \pi 2 + \frac{\sigma -\pi}{2}, \frac \pi 2 - \frac{\sigma - \pi}{2}]$,
since otherwise there would be a shorter arc containing all the eigenvalues.
It is easy to see by drawing a parallelogram on the complex plane 
that the magnitude of $e^{\ii \sigma/2} + e^{-\ii \sigma/2} + e^{\ii x}$ for varying~$x \in J$ assumes the largest value~$1$
when $x$ is either end point of the interval~$J$.
So,
\begin{equation}
    \abs*{\frac 1 \dd \sum_i \lambda_i } \le \frac 1 \dd \abs*{\lambda_1 + \lambda_j + \lambda_\dd} + \frac 1 \dd \abs*{\sum_{i \neq 1, j,\dd} \lambda_i} \le \frac 1 \dd + \frac{\dd - 3}{\dd} = \frac{\dd - 2}{\dd}.
\end{equation}
Therefore, if we set $\sigma' = \min(\sigma,\pi)$ in all cases, we have
\begin{align}
    \abs*{\frac 1 \dd \sum_i \lambda_i } &\le \frac 2 \dd \cos \frac {\sigma'} 2 + \frac{\dd -2}{\dd} ,\\
    \ol{F}(\calU(U),\calU(V)) 
    &\ge 1 - \left( \frac 2 \dd \cos \frac {\sigma'} 2 + \frac{\dd -2}{\dd}  \right)^2 \ge \frac 2 \dd \sin^2 \frac {\sigma'} 2 = \frac 1 {2\dd}  \diamnorm*{ \calU(U) - \calU(V) }^2 \nonumber
\end{align}
where the second inequality in the second line uses $\dd \ge 2$ and $\cos (\sigma'/2) \ge 0$.

The tightness of the left-hand side is seen by setting $\lambda_1 = \cdots = \lambda_{\dd/2} \neq \lambda_{\dd/2 +1} = \cdots =\lambda_\dd$
where $\dd$ is even.
The tightness of the right-hand side is seen by setting 
$\lambda_1 = e^{\ii \sigma / 2}$, $\lambda_\dd = e^{-\ii\sigma/2}$ with small~$\sigma$ and $\lambda_i = 1$ for all other~$i$.
\end{proof}

\yrcmixing*

The inequality follows from \cite[Lem.~2]{Yang2020},
but we give another short proof (with worse constants) using the idea of ``Mixing Unitary'' Lemma of~\cite{Hastings2016,Campbell2016}.

\begin{proof}[Proof of \cref{mixingIneq}]
    Write $Z^\dag \bU = e^\bX$ where $\opnorm{\bX} \le \pi$.
    Since the distribution of~$\bX$ is conjugation invariant,
    \begin{equation}
        \E_\bX \bX = \E_\bX \int \rd V ~ V \bX V^\dag = \E_\bX \frac{\Tr \bX}{\dd} \id  .
    \end{equation}
    Then, for any state~$\rho$ that may be entangled with some ancilla,
    we see, using nested commutators 
    $[\bX,\rho]_k = [\bX,[\bX,\rho]_{k-1}]$ for $k > 0$ and $[\bX,\rho]_0 = \rho$,
    \begin{multline}
        \diamnorm{\calM(\bU) - \calU(Z)}
        = \onenorm*{
                \E_{\bU} \bU \rho \bU^\dag 
                -
                Z \rho Z^\dag
                }
            =
        \onenorm*{            
                \E_\bX e^\bX \rho e^{-\bX} - \rho
            }
            =
            \onenorm*{
            \E_\bX \sum_{k=1}^{\infty} \frac{[\bX,\rho]_k}{k!}
            } \\
            \le
            \onenorm*{
                \E_\bX [\bX,\rho]
            }
             +
         \onenorm*{
             \E_\bX \sum_{k=2}^{\infty} \frac{[\bX,\rho]_k}{k!}
         }
         \le
             \E_\bX \sum_{k=2}^\infty \frac{2^k \opnorm{\bX}^k \onenorm{\rho}}{k!} \\
         \le
             \E_\bX 4 \opnorm{\bX}^2 \exp(2 \pi)
         =
             \E_\bU 4 \dist(Z, \uone \bU)^2 \exp(2 \pi).
    \end{multline}
    The result follows using norm equivalence (\cref{prop:norm-equivalence}).
\end{proof}

\section{Estimating eigenvalues} \label{sec:eigenvalues}

In this section we describe an algorithm that, given access to an unknown unitary~$Z \in \UU(\dd)$, outputs an approximation to the eigenvalues of~$Z$ that is close in ``$L_\infty$-distance up to a global phase''.
Let us make this precise:
\begin{definition}
    Let $\Theta_1, \Theta_2 \subset [0,2\pi)$ be two finite sets of ``eigenphases''.
    We write
    \begin{equation}
        \dH{\Theta_1}{\Theta_2} = \min_{\tau \in [0,2\pi)} d_{\mathrm{H}}(\Theta_1 + \tau, \Theta_2),
    \end{equation}
    where $d_{\mathrm{H}}$ denotes Hausdorff distance in~$\RR/2\pi\ZZ$.
    In other words, $\dH{\Theta_1}{\Theta_2} \leq \eps$ iff, for some $\tau$, the sets $\Theta_1' = \Theta_1 + \tau$ and $\Theta_2$ have the following property (when all numbers are taken mod~$2\pi$): For every $\theta_1 \in \Theta'_1$ there is $\theta_2 \in \Theta_2$ with $\abs{\theta_1 - \theta_2} \leq \eps$, and vice versa (interchanging the roles of $\Theta'_1$~and~$\Theta_2$). 
\end{definition}

Our goal is to prove the following theorem:
\begin{theorem} \label{thm:esteigs}
    There is an eigenvalue estimation algorithm that, given $\eps > 0$ and access to an unknown unitary $Z \in \UU(\dd)$, applies $Z$ at most $O(\dd/\eps)\cdot \log^2 \dd$ times and outputs the classical description of a set~$\wh{\Theta} \subset [0,2\pi)$ such that $\dH{\Theta}{\wh{\Theta}} \leq \eps$ except with probability at most~$\dd^{-100}$, where $\Theta$ is the set of $Z$'s eigenphases.  (That is, $\Theta = \{ \theta \in [0,2\pi) : e^{\ii \theta} \text{ is an eigenvalue of } Z\}$.)
    Moreover, the quantum space requirement for the algorithm is only $2$~qudits plus $1$-qubit, and the gate complexity beyond the uses of~$Z$ is only $\poly\log(\dd/\eps)$.
\end{theorem}

The overall algorithm is similar to Phase Estimation. 
We begin as follows:
\begin{proposition} \label{prop:est1}
    There is an algorithm that, given $\eps,\eta > 0$, access to an unknown $Z \in \UU(\dd)$, and two eigenvectors $\ket{a}$, $\ket{b}$ with eigenphases $\alpha, \beta \in [0,2\pi)$ (respectively), has the following behavior:  
    The algorithm uses one additional qubit, applies $Z$ at most $O(\log(1/\eta)/\eps)$ times, and outputs an estimate~$\theta$ that is within~$\pm \eps$ of $\beta-\alpha$, except with probability at most~$\eta$.
    In addition, at the end of the algorithm still holds $\ket{a} \otimes \ket{b}$, unentangled with its results.
\end{proposition}
\begin{proof}
    The algorithm strongly resembles Quantum Phase Estimation.
    Given $\ket{a} \otimes \ket{b}$ in ``registers $1$ and~$2$'', suppose we:
    \begin{itemize}
        \item Adjoin a qubit $\ket{+}$ in a new ``register~$0$''.
        \item Performed controlled-SWAP on registers $1,2$, controlled on register~$0$.
        \item Apply $Z$ to register~$1$.
        \item Perform controlled-SWAP again.
        \item Detach the qubit in register~$0$.
    \end{itemize}
    It is easy to calculate that this procedure leaves registers $1$~and~$2$ in the state $\ket{a}\otimes \ket{b}$, with the detached qubit in state $\sqrt{1/2} (e^{\ii \alpha}\ket{0} + e^{\ii \beta} \ket{1})$.
    Up to a global phase, this is $\sqrt{1/2}(\ket{0} + e^{\ii(\beta-\alpha)}\ket{1})$.
    We can then repeat this procedure but with $Z^2$, $Z^4$, $Z^8$, etc.\ in place of~$Z$, yielding qubits $\sqrt{1/2}(\ket{0} + e^{\ii 2^k(\beta-\alpha)}\ket{1})$ for all $0 \leq k < m$, at a cost of $2^m - 1$ uses of~$Z$.
    This is precisely the scenario arising within textbook Quantum Phase Estimation~\cite{cleve1998quantum}, prior to its QFT.
    If we were to finish with QFT, we would essentially complete the proof, except that it would use $m = \log(1/\eps) + O(1)$ qubits to achieve error~$\eps$ with probability~$2/3$.
    Instead, if we use Iterative Quantum Phase Estimation~\cite{dobvsivcek2007arbitrary}, we only need one additional qubit to get the same guarantee.
    Finally, repeating the whole algorithm $O(\log (1/\eta))$ times and taking the median result completes the proof.
\end{proof}

\begin{proposition} \label{prop:est2}
    There is an algorithm that, given $\eps > 0$, access to an unknown $Z \in \UU(\dd)$, and an eigenvector $\ket{a}$, has the following behavior:  
    The algorithm uses one additional qudit and one additional qubit, applies $Z$ at most $O(\dd/\eps) \cdot \log^2 \dd$ times, and outputs the classical description of a set~$\wh{\Theta} \subset [0,2\pi)$ such that $\dH{\Theta}{\wh{\Theta}} \leq \eps$ except with probability at most~$\dd^{-100}$, where $\Theta$ is the set of $Z$'s eigenphases. 
\end{proposition}
\begin{proof}
    Write $\alpha$ for the eigenphase of~$Z$ on $\ket{a}$.
    The algorithm repeatedly does the following: 
    Adjoin to~$\ket{a}$ a second qudit in the maximally mixed state.
    Then use the routine from \Cref{prop:est1}, with $\eta = \dd^{-C}$ for some large constant~$C$.
    The result is that except with probability at most~$\dd^{-C}$, the routine uses $O(1/\eps) \cdot \log d$ applications of~$Z$, and ends up holding the following: $\ket{a} \otimes \ket{\bb}$ together with an $\eps$-accurate estimate of $\boldsymbol{\beta}-\alpha$, where $\ket{\bb}$ is a uniformly random eigenvector of~$Z$ and $e^{\ii \boldsymbol \beta}$ is the associated eigenvalue.
    At this point, the difference $\boldsymbol{\beta}-\alpha$ may be recorded, and the register containing $\ket{\bb}$ discarded.
    
    By repeating this procedure $O(\dd \log \dd)$ times, the Coupon Collector analysis ensures that the algorithm will record $\eps$-accurate values of $\theta - \alpha$ for the at most~$d$ distinct values~$\theta \in \Theta$.
    (Except with probability at most~$\dd^{-100}$, having chosen~$C$ appropriately.)
    It follows that the collection $\wh{\Theta}$ of recorded values satisfies $\dH{\Theta}{\wh{\Theta}} \leq \eps$ (independent of what~$\alpha$ is), as required.
\end{proof}
\begin{remark}
    If $\log(1/\eps) \ll \log \dd$ then the log factors in the preceding analysis may be slightly improved. 
    In this case, since the number of possible values for $\theta - \alpha$ is only $O(1/\eps)$, one can use analysis of  Non-Uniform Coupon Collecting to show that only $O(\dd \log 1/\eps)$ repetitions are required. 
    We omit further details.  
\end{remark}

Finally, \Cref{thm:esteigs} now follows by applying the algorithm from \Cref{prop:est2} with the maximally mixed state in place of~$\ket{a}$.

\section{Gate-efficient pure state tomography with optimal sample complexity} \label{app:gate-state}

Although it doesn't seem to significantly improve the gate complexity of our algorithm, we record the following result:
\begin{theorem}
    There is a tomography algorithm for $\dd$-dimensional pure states, $\dd = 2^\nn$,  with the following behavior.  Given $\eps > 0$ and $O(\dd/\eps)$ copies of an unknown $\nn$-qubit pure state $\ket{u} \in \CC^{\dd}$, the algorithm sequentially measures the copies using von Neumann measurements and collects the classical results.  The measurements are nonadaptively chosen using classical randomness, and each is implemented with gate complexity $\nn^{O(1)}\cdot \log(1/\eps)$ on the minimal number of qubits,~$\nn$.  
    Finally, after the classical results are processed in $d^{O(1)}$ time, an estimate~$\ket{\widehat{\bu}}$ is output, and except with probability at most~$\dd^{-100}$ it satisfies $\abs{\braket{u|\widehat{\bu}}} \geq 1 - \eps$.
\end{theorem}
We were unable to find this theorem in the literature, although it follows from relatively standard ideas.
It would seem that all previous works on pure state tomography (e.g.,~\cite{gross2010quantum,CL14,KUENG201788,Guta2020}) either have optimal sample complexity but are not gate-inefficient, or are gate-efficient but have sample complexity no better than $O(\dd \log(\dd)/\eps)$.
\begin{proof}
    Recall the standard pure state tomography algorithm from \Cref{prop:puretomog} and its analysis from \Cref{sec:ancillary}.
    Slightly more strongly we use that, with $m = O(\dd/\eps)$ copies, it can be made to satisfy (\cite[Theorem~5]{Guta2020})
    \begin{equation}
        \Pr[\opnorm{\ket{u}\!\bra{u} - \bL} \geq t\sqrt{\eps}] \leq 2 \exp(O(\dd) \cdot (1 - t^2))
    \end{equation}
    simultaneously for all constants $t > 0$, from which 
    \begin{equation}
        \Ex[\opnorm{\ket{u}\!\bra{u} - \bL}] \leq c \sqrt{\eps}
    \end{equation}
    easily follows (for some constant~$c$).
    Hence for any positive even integer~$k$,
    \begin{equation}    \label{eqn:trc}
        0 \leq \Ex[\Tr((\ket{u}\!\bra{u} - \bL)^{k})] \leq \dd(c \sqrt{\eps})^{k}.
    \end{equation}
    Recall that
    \begin{equation}
        \bL = \frac{\dd+1}{m}\sum_{t=1}^m \bV_t \ket{\bJ_t}\!\bra{\bJ_t}\bV_t^\dagger  - \id,
    \end{equation}
    where $\bV_1, \dots, \bV_m$ are independent Haar-random unitaries and $\bJ_t$ denotes the column index of the $t$th measurement outcome.
    For a given $t \in [m]$, let us define a ``$\bV_t$-entry'' to be an expression of the form $\braket{i|\bV_t|\bJ_t}\braket{\bJ_t|\bV_t^\dagger|i'}$, where $\ket{i}, \ket{i'}$ are standard basis vectors.
    
    Now
    \begin{align*}
        &\phantom{{}=} \Ex[\Tr((\ket{u}\!\bra{u} - \bL)^{k})] 
        = \Ex\bracks[\Bigg]{\Tr\parens[\Big]{\parens[\Big]{\id + \ket{u}\!\bra{u} - \frac{\dd+1}{m}\sum_{t=1}^m \bV_t \ket{\bJ_t}\!\bra{\bJ_t}\bV_t^\dagger}^{k}}} \\
        &= \Ex\bracks[\Bigg]{\Tr\parens[\Big]{\parens[\Big]{\sum_{i_1=1}^d \ket{i_1}\!\bra{i_1} + \ket{u}\!\bra{u} - \frac{\dd+1}{m}\sum_{t=1}^m \sum_{i_2,i_3=1}^{\dd}\ket{i_2}\!\braket{i_2|\bV_t|\bJ_t}\!\braket{\bJ_t|\bV_t^\dagger|i_3}\!\bra{i_3}}^{k}}},
    \end{align*}
    which in turn is a sum of $(\dd+1+m\dd^2)^k$ ``monomials'', each of the form
    \begin{equation}    \label{eqn:monom}
        c\cdot (\text{product of $k_1$ $\bV_1$-entries})(\text{product of $k_2$ $\bV_2$-entries})\cdots(\text{product of $k_{m}$ $\bV_m$-entries})
    \end{equation}
    for a constant $c$ with $|c| \leq 1$ and nonnegative integers~$k_t$ summing to at most~$k$.
    Note that the expected value of a monomial as in \Cref{eqn:monom} is equal to
    \begin{equation}
        c \cdot \E[(\text{product of $k_1$ $\bV_1$-entries})] \cdots \E[(\text{product of $k_{m}$ $\bV_m$-entries})],
    \end{equation}
    since $\bV_1, \dots, \bV_m$ are independent.
    
    Let us now introduce independent but \emph{pseudorandom} unitaries $\bV'_1, \dots, \bV'_m$.
    We write~$\bL'$ and $\bJ'_t$ and repeat the above development.
    Now suppose the following holds: 
    \begin{equation}    \label{eqn:achieve}
        \abs*{\E[(\text{product of $k_t$ $\bV'_t$-entries})] - \E[(\text{product of the same $k_t$ $\bV_t$-entries})]} \leq \eta \qquad \text{for any $k_t \leq k$}.
    \end{equation}
    Then---using that any $\bV_t$-entry or $\bV'_t$-entry has magnitude at most~$1$---it's not hard to conclude that
    \begin{equation}
        \abs*{\Ex[\Tr((\ket{u}\!\bra{u} - \bL')^{k})] - \Ex[\Tr((\ket{u}\!\bra{u} - \bL)^{k})]} \leq (\dd + 1 + m\dd^2)^k \cdot m\eta.
    \end{equation}
    Thus for some $\eta \leq (\dd/\eps)^{-O(k)}$ the difference is at most~$\eps^k$ and from \Cref{eqn:trc} we get 
    \begin{equation}
        0 \leq \Ex[\Tr((\ket{u}\!\bra{u} - \bL')^{k})] \leq 2\dd(c \sqrt{\eps})^{k} \quad\implies\quad \Pr[\Tr((\ket{u}\!\bra{u} - \bL')^{k}) \geq 2\dd^{101}(c \sqrt{\eps})^{k}] \leq \dd^{-100}.
    \end{equation}
    But when $\Tr((\ket{u}\!\bra{u} - \bL')^{k}) \geq 2\dd^{101}(c \sqrt{\eps})^{k}$ we have
    \begin{equation}
        \opnorm{\ket{u}\!\bra{u} - \bL'} \leq (2\dd^{101})^{1/k}(c\sqrt{\eps}) = O(\sqrt{\eps}), 
    \end{equation}
    provided $k = \Omega(\log d)$.

    In summary, provided:
    \begin{itemize}
        \item $\bV'_1, \dots, \bV'_m$ achieve \Cref{eqn:achieve} for $k = \Omega(\log d)$ and $\eta = O(\dd/\eps)^{-O(k)}$,
        \item each $\bV'_t$ can be implemented with gate complexity $\nn^{O(1)} \cdot \log(1/\eps)$,
    \end{itemize} 
    we get that $\opnorm{\ket{u}\!\bra{u} - \bL'} \leq O(\sqrt{\eps})$ except with probability at most~$\dd^{-100}$, from which the result follows as in \Cref{sec:ancillary}.

    Let us now consider a generic product of $k_t$ $\bV_t$-entries.  Dropping the $t$~subscript for notational simplicity, it looks like 
    \begin{equation}
        \braket{i_1|\bV|\bJ}\!\braket{\bJ|\bV^\dagger|\ell_1} \cdot 
        \braket{i_2|\bV|\bJ}\!\braket{\bJ|\bV^\dagger|\ell_2} \cdots
        \braket{i_k|\bV|\bJ}\!\braket{\bJ|\bV^\dagger|\ell_k}.
    \end{equation} 
    Since $\Pr[\bJ = j \mid \bV = V] = \braket{u|V|j}\!\braket{j|V^\dagger|u}$, the expectation of the above quantity is 
    \begin{equation}
        \sum_{j=1}^\dd \Ex[\braket{u|\bV|j}\!\braket{j|\bV^\dagger|u}\cdot \braket{i_1|\bV|j}\!\braket{j|\bV^\dagger|\ell_1} \cdot 
        \braket{i_2|\bV|j}\!\braket{j|\bV^\dagger|\ell_2} \cdots
        \braket{i_k|\bV|j}\!\braket{j|\bV^\dagger|\ell_k}].
    \end{equation} 
    Since $\ket{u}$ is a unit vector, it follows that the above expectation changes by no more than~$\sqrt{\dd} \cdot \lambda$ in magnitude if $\bV$ is replaced by~$\bV'$ distributed as a quantum $(\dd, \lambda, k+1)$-tensor-product-expander (see Definition~1 and equations (4)--(6) of~\cite{brandao2016local}).
    By the work of~\cite{brandao2016local} (see also the latest strengthening from~\cite{haferkamp2022random}), we know that a $(\dd, \lambda, k)$-TPE can be generated by certain simple probability distribution on $\nn$-qubit unitary circuits with $\mathrm{poly}(\nn,k) \cdot\log(1/\lambda)$ gates (provided $k = o(\dd)$).  
    With $k = \Theta(\log \dd) = \Theta(\nn)$ and $\lambda = O(\dd/\eps)^{O(k)}$, this is $\nn^{O(1)} \log(1/\eps)$ gates, as needed to complete the proof of the theorem.
\end{proof}

\section{\texorpdfstring{Compressed proof of~\cref{ChannelLearningConfidenceBound}}{Compressed proof of Proposition 4.2}}\label{app:Bavaresco}

Here we give a summary of the proof of \cite[Thm.~5]{Bavaresco2021}.
The content in this appendix is an excerpt of results in~\cite{Chiribella2007, Chiribella2016, Hashimoto2009,  Bavaresco2021}.
Our exposition will omit much of general discussion in those references,
but present necessary pieces.

\paragraph{Step~1: Quantum Testers.}

It is well known that a quantum channel~$\calC$ (unitary or not)
corresponds to a Choi operator~$C$,
which is nothing but a density operator obtained by applying~$\calC$
to an unnormalized pure density matrix 
$\sum_{i,i'} \ket{ii}\bra{i'i'}$ 
that is maximally entangled with an ancilla
of dimension equal to that of the input of~$\calC$.
The Choi operator of the composition of quantum channels
has a formula in terms of the Choi operators of the component channels,
using so-called link product denoted by~$*$~\cite{Chiribella2007}.
Although the link product involves \emph{partial transposes} 
(that generally do not preserve positivity)
whenever an output of a channel
is plugged in to another channel,
the result of a link product is always a positive semidefinite (PSD) operator.

Let us call a sequence of quantum channels a quantum network,
which is not necessarily a unitary.
It has been observed~\cite{Chiribella2007}
that even if some part of a quantum network is unspecified,
one can associate a Choi operator to a quantum network 
by introducing a Hilbert space
for each of input and output of an unspecified channel,
giving a blank slot that will be filled by a quantum channel.
Given a Choi operator~$T_x$ for a quantum network 
with $Q$ slots unfilled and with an outcome~$x$ postselected,
if we have $Q$ Choi operators~$C_1,C_2,\ldots,C_Q$ of quantum channels 
to fill the $Q$ slots,
we can express the probability of obtaining~$x$ as
\begin{align}
\Pr[x] = T_x * C = \Tr( T_x C^T ) = \Tr( T_x^T C ) \label{eq:linkproduct}
\end{align}
where 
$C = C_1 \otimes C_2 \otimes \cdots \otimes C_Q$
and 
$*$ is the link product~\cite[\S 2.4]{Chiribella2016}.
This encompasses the scenario where the $Q$ channels are not parallel but causally ordered
since $T_x$ can include \emph{e.g.} identity channels in between the slots.
Here, $C^T$ is the full, rather than a partial, transpose of~$C$.
Normally, a link product of two Choi operators
takes partial transpose on the input Hilbert space of the causally succeeding factor,
but since each channel~$C_j$ takes input from an open ``leg'' of~$T_x$ 
and some other leg of~$T_x$ takes input from the output of~$C_j$,
we equivalently take full transpose of~$C_j$ while keeping~$T_x$ intact.
We may think of this as a consequence of the fact that $T_x * C$ is  
a one-dimensional operator, a number.
Note that $C_j^T$ is a Choi operator of a quantum channel;
in particular, if $C_j$ is that of a unitary channel, so is $C_j^T$.
The collection $\{ T_x \}_x$ is called a quantum tester~\cite{Chiribella2009}.

\paragraph{Step~2: Semidefinite programming.}

Now we consider the problem in the statement of~\cref{ChannelLearningConfidenceBound}.
There are $N$ candidate channels, labeled by $x=1,\ldots,N$.
We are given access to $Q$ uses of one of the candidates.
Letting a quantum network to output~$y$,
we have a tester $\{T_y~|~ y = 1,2,\ldots,N \}$
that has $Q$ blank slots to which we plug in $Q$ identical channels,
chosen uniformly at random from the candidates.
Given $Q$ uses of a candidate channel labeled by~$x$,
the probability of outputting~$y$ is $\Tr(T_y (V_x^{\otimes Q})^T)$ by~\cref{eq:linkproduct}
where $V_x$ is the corresponding Choi operator.
The average probability of outputting a correct label
is given by
\begin{align}
    p_{succ} = \frac 1 N \sum_{x=1}^N \Tr(T_x (V_x^{\otimes Q})^T). \label{eq:prsucc}
\end{align}
Now, each operator~$T_x$ is a link product of some PSD operators,
and hence is PSD.
In addition, since probabilities must add up to~$1$,
we have $\Tr(\sum_y T_y (V_x^{\otimes Q})^T) = 1$ for each~$x$.

Actually, the blank slots of a tester may be filled with arbitrary channels
and must still produce a probability distribution.
In particular, we must have
\begin{align}
\Tr(T_\text{all} V^{\otimes Q}) = 1, \quad T_\text{all} = \sum_x T_x
\end{align}
for \emph{every} $V$ that is 
the Choi operator of a unitary channel.
Here we do not need to take transpose 
because the transpose 
of the Choi operator~$V(U)$ of a unitary channel~$\calU(U)$
is still the Choi operator~$V(U^*)$ of the complex conjugate unitary channel~$\calU(U^*)$.
This condition consists of infinitely many equations;
however, one observes~\cite{Chiribella2016}
that this condition is equivalent to demanding $\Tr(T_\text{all} W) = 1$
for any affine combination~$W$ in the affine span
\begin{align}
\calW = \left\{ \sum_i \xi_i V(U_i)^{\otimes Q} ~\middle|~ U_i \in \UU(\dd), \xi_i \in \RR, \sum_i \xi_i = 1 \right\}
\end{align}
This seemingly more complicated condition 
simplifies the situation 
because any affine space in a finite dimensional vector space 
has a finite basis. 
So, the condition transcribes to a finite number of equations~$\Tr(T_\text{all} W_j) = 1$,
one for each affine basis element~$W_j$ of~$\calW$.

A set of PSD operators $T_y$ 
with the affine constraint $\Tr(T_\text{all} W_j) = 1$
may go beyond what a physical quantum network gives
and may not qualify as a tester;
we do not claim that physically realizable testers 
are fully characterized by~$\calW$.
Nonetheless, we consider a semidefinite program~\cite{Chiribella2016}
specified by~$\{ V_x ~|~ x=1,2,\ldots,N \}$:
\begin{align}
\text{maximize } & \frac 1 N \sum_{x=1}^N \Tr(T_x (V_x^{\otimes Q})^T) 
\text{ by varying $\{T_x\}_x$ subject to } 
\begin{cases}
    T_y \succeq 0 \quad \forall y,\\
    \Tr(T_\text{all} W_j) = 1 \quad \forall j.
\end{cases} \label{primarySDP}
\end{align}
Since any physically realizable tester is a feasible solution,
the optimal value of this program is an upper bound on~$p_{succ}$.
Suppose that $\{T_x\}_x$ is a feasible solution.
Observe that 
\begin{align}
 \exists \lambda \in \RR,~~
 \exists W  \in \calW,~~
 \forall y:  ~\frac 1 N (V_y^{\otimes Q})^T \preceq \lambda W
 \qquad
 \Longrightarrow
 \qquad
 p_{succ} \le \lambda
 \label{eq:lambdaCondition}
\end{align}
because, if $W = \sum_j \xi_j W_j$ with $\sum_j \xi_j = 1$, then
\begin{align}
    \sum_x \frac 1 N \Tr(T_x (V_x^{\otimes Q})^T) 
    \le 
    \sum_x \Tr(T_x \sum_j \lambda \xi_j W_j) 
    = \lambda \sum_j  \xi_j \Tr(T_\text{all} W_j) 
    = \lambda .
\end{align}
It is an instance of the weak duality of semidefinite programming.
The larger the class~$\calW$ is, the stronger our bound $\lambda$ will be.

\paragraph{Step~3: Some representation theory.}

By construction, $\calW$ contains
\begin{align}
    \bar C = \int_{\UU(\dd)} V(U)^{\otimes Q} \rd U
\end{align}
where $\rd U$ is the Haar measure on~$\UU(\dd)$ and $V(U)$ is the Choi operator
of the unitary channel~$\calU(U)$.
We claim~\cite{Bavaresco2021} that for any $T \in \UU(\dd)$,
\begin{align}
    V(T)^{\otimes Q} \preceq \underbrace{\binom{ Q + \dd^2 - 1 }{Q}}_{= N \lambda}  \bar C, \label{eq:VC}
\end{align}
which is in the form of~\cref{eq:lambdaCondition}
and will therefore complete the proof of~\cref{ChannelLearningConfidenceBound}.
Now, \cref{eq:VC} is equivalent to saying~\cite{Hashimoto2009}
that there is $\gamma = \binom{Q+\dd^2-1}{Q} > 0$ 
with which $\abs{\braket{\phi|\psi}}^2 
\le \gamma \E_U \bra \phi R(U) \ket \psi \bra \psi R(U)^\dag \ket \phi$ 
for any vector $\ket \phi$ where
$\ket \psi$ is a pure unnormalized state corresponding to~$V(T)^{\otimes Q}$
and
$R(U) = (\id \otimes U)^{\otimes Q}$ 
is a \emph{unitary representation} of a compact Lie group~$\UU(\dd)$.
It is standard to work in a basis where $R(U)$ is block-diagonal,
so $R(U) = \bigoplus_\mu R_\mu(U) \otimes \id_{m(\mu)}$ 
where $\mu$ ranges over all inequivalent irreps occurring in~$R$.
The natural number~$m(\mu)$ is the multiplicity of~$\mu$ within~$R$.
Correspondingly we have $\ket \psi = \bigoplus_\mu \ket{\psi_\mu}$
and $\ket \phi = \bigoplus_\mu \ket{\phi_\mu}$
where each orthogonal summand, $\ket{\psi_\mu}$ or $\ket{\phi_\mu}$,
can be viewed as an entangled state across the irrep $R_\mu$ 
and the ``multiplicity space.''
Let~$d_\mu = \dim R_\mu$ and $d'_\mu = \min(d_\mu, m(\mu) )$.
We may write $\ket{\psi_\mu} = \sum_{a=1}^{d'_\mu} \ket{\psi^R_{\mu,a}}\ket{\psi^M_{\mu,a}}$
where each $\ket{\psi^R_{\mu,a}}$ is normalized.
By Cauchy--Schwarz,
\begin{align}
\abs{\braket{\phi|\psi}}^2 
&= 
\abs*{\sum_\mu \sum_{a(\mu)} \sqrt{\frac{d_\mu}{d_\mu}} \braket{\phi_{\mu}|\psi^R_{\mu,a}\psi^M_{\mu,a}} }^2
\le
\underbrace{\left( \sum_\mu  \sum_{a(\mu)} d_\mu \right)}_{\gamma'}
\left( \sum_\mu \sum_{a(\mu)} \frac 1 {d_\mu} \abs{\braket{\phi_{\mu}|\psi^R_{\mu,a}\psi^M_{\mu,a}} }^2 \right).
\end{align}
On the other hand,
the average of conjugation by an irrep~$R_\mu$ 
is to trace it out and replace it by the identity, 
scaled to preserve the trace (Schur's lemma). Hence, 
\begin{align}
 \E_U \bra \phi R(U) \ket \psi \bra \psi R(U)^\dag  \ket \phi 
= \sum_\mu \bra{\phi_\mu} \frac{1}{d_\mu} \id_{R_\mu} \otimes \psi^M_\mu \ket{\phi_\mu}
\ge 
\sum_\mu \sum_{a(\mu)} \frac 1 {d_\mu} \abs{\braket{\phi_{\mu}|\psi^R_{\mu,a}\psi^M_{\mu,a}} }^2 .
\end{align}
Therefore, we have proved $V(T)^{\otimes Q} \preceq \gamma' \bar C$~\cite{Hashimoto2009}.

It remains to show that $\gamma' \le \binom{Q + \dd^2 -1}{Q} = \gamma$.
Because $d'_\mu \le d_\mu = \dim R_\mu$ by definition, 
we aim to show
\begin{align}
    \sum_\mu (\dim R_\mu)^2 = \binom{Q + \dd^2 -1}{Q}. \label{eq:schurthesis}
\end{align}
As noted in~\cite{Bavaresco2021}, 
this is a formula appearing in Schur's thesis~\cite[Eq.~(57)]{Schur}.
We note the following pointers to modern textbooks (\emph{e.g.}~\cite[App.~A]{FultonHarris})
to prove this identity.
Obviously, an irrep  appears in~$R(U) = (\id \otimes U)^{\otimes Q}$
if and only if it does in~$U^{\otimes Q}$.
It is well known that 
the irreps $R_\mu$ of~$\UU(\dd)$ 
in the $Q$-fold tensor representation correspond to Young diagrams~$\mu$
with~$Q$~boxes and at most~$\dd$~rows (English notation).
The character~$\Tr R_\mu(U)$ is given 
by the Schur polynomial~$s_\mu(u_1,\ldots,u_\dd)$ of degree~$Q$
evaluated at the eigenvalues~$u_k$ of~$U$.
An identity that is useful for us is the Cauchy identity
\begin{align}
    \sum_\lambda s_\lambda(u) s_\lambda(v) = \prod_{k=1}^\dd \prod_{\ell = 1}^\dd \frac 1 { 1 - u_k v_\ell }
    \label{eq:Cauchy}
\end{align}
where $\lambda$ ranges over \emph{all} Young diagrams 
with at most~$\dd$~nonzero rows but unlimited number of boxes.
Here all the inverse polynomials should be 
understood as a formal power series.
The dimension~$d_\mu$ is~$s_\mu(1,1,\ldots,1)$,
so if we take the degree~$2n$ part of the right-hand side of~\cref{eq:Cauchy}
and evaluate it at $u_k = v_\ell = 1$ for all~$k,\ell$,
then we obtain~$\sum_\mu (\dim R_\mu)^2$.
This is equivalent to reading off the coefficient of~$t^Q$
in the series expansion of~$(1-t)^{-\dd^2}$.
Applying $(Q!)^{-1} \partial_t^Q$, 
the identity~\cref{eq:schurthesis} follows.
\end{document}